\documentclass[aps,superscriptaddress,nofootinbib]{revtex4}
\usepackage{shortcuts}
\usepackage{color}
\usepackage{datetime}

\def\red{\color{red}}
\def\black{\color{black}}
\def\blue{\color{blue}}
\numberwithin{equation}{section}
\usepackage{subfigure}
\usepackage{hyperref}
\usepackage[utf8]{inputenc}
\usepackage{tikz,tikz-3dplot}
\usepackage{comment}
\usetikzlibrary{decorations.markings}
\usepackage{mathtools}
\usepackage{leftidx}

\def\strong{\mathsf{STR}}
\def\BQP{\mathsf{BQP}}
\def\PP{\mathsf{PP}}

\begin{document}
\title{Quantum circuit dynamics via path integrals:\\
 Is there a classical action for discrete-time paths?}
\author{Mark D. Penney}
\email{mark.penney@maths.ox.ac.uk}
\affiliation{Mathematical Institute, University of Oxford, Woodstock Road, Oxford, UK OX2 4GG}
\author{Dax Enshan Koh}
\email{daxkoh@mit.edu}
\affiliation{Department of Mathematics, Massachusetts Institute of Technology, Cambridge, Massachusetts 02139, USA}
\author{Robert W. Spekkens}
\email{rspekkens@perimeterinstitute.ca}
\affiliation{Perimeter Institute for Theoretical Physics, 31 Caroline Street North, Waterloo, Ontario, Canada N2L 2Y5}

\date{\today\ }

\begin{abstract}
It is straightforward to compute the transition amplitudes of a quantum circuit using the sum-over-paths methodology when the gates in the circuit are balanced,
where a balanced gate is one for which all nonzero transition amplitudes are of equal magnitude.
Here we consider the question of whether, for such circuits, the relative phases of different discrete-time paths through the configuration space can be defined in terms of a classical action, as they are for continuous-time paths.  
We show how to do so
 for certain kinds of quantum circuits, namely, Clifford circuits where the elementary systems are continuous-variable systems or discrete systems of odd-prime dimension.
These types of circuit are distinguished by having phase-space representations that serve to define their classical counterparts.  For discrete systems, the phase-space coordinates are also discrete variables. 
We show that for each gate in the generating set, one can associate a symplectomorphism on the phase-space and to each of these one can associate a generating function, defined on two copies of the configuration space.  For discrete systems, the latter association is achieved using tools from algebraic geometry.
Finally, we show that if the action functional for a discrete-time path through a sequence of gates is defined using the sum of the corresponding generating functions, then it yields the correct relative phases for the path-sum expression.  These results are likely to be relevant for
 quantizing physical theories where time is fundamentally discrete, characterizing the classical limit of discrete-time quantum dynamics, and proving complexity results for quantum circuits.
\end{abstract}
\maketitle
\tableofcontents

\section{Introduction and Summary}

\subsection{The question}

The sum-over-paths methodology in quantum mechanics, pioneered by Richard Feynman, offers an alternative to the standard means of expressing quantum dynamics, just as the least-action formulation of classical dynamics offers an alternative to the standard Hamiltonian formulation \cite{FeynmanHibbs}. 
  In particular, it allows one to determine the probability amplitude of making a transition among states for any given (possibly time-dependent) Hamiltonian operator describing the quantum dynamics of the system.
There is, however, a second type of problem to which it can be applied. Here, one is given a \emph{modular} description of the quantum system's dynamics---for instance, a description of a quantum circuit with gates that are drawn from some fixed set of possibilities---and the goal is to compute the transition amplitudes of the overall circuit from a knowledge of the transition amplitudes of each gate.  

The distinction between these two types of problems is best illustrated by an example.
Suppose one is interested in the transverse position of an atom as it passes through an interferometer.  It is then useful to treat different components in the interferometer as gates in a circuit. Determining the propagator associated to a particular gate given a knowledge of the Hamiltonian governing the dynamics of the atom as it passes through that gate is a problem of the first sort.  Determining the propagator associated to the entire interferometric set-up given a knowledge of the propagators associated to each gate is a problem of the second sort.  We shall refer to the two sorts of problems henceforth as the {\em continuous-time scenario} and the {\em circuit scenario} respectively. 
 In either scenario, one can consider the system's degrees of freedom to be discrete or continuous.  An interferometer is an example of a circuit acting on continuous degrees of freedom, while the circuits that are most commonly studied in the field of quantum computation involve discrete degrees of freedom. 

A quantum circuit can be specified by a sequence of gates, where each gate is characterized by a unitary operator. 
The dynamics that occurs \emph{within} each gate is generally not specified.  This is because only the overall functionality of the gate is important for the functionality of the circuit as a whole, and there are many different choices of the dynamics within the gate that lead to the same functionality.  For instance, a piece of polaroid and a birefringent crystal both allow one to achieve the overall functionality of a polarization filter, even though the evolution of the light within the two sorts of components is quite different.  Given that the dynamics internal to each gate is irrelevant---and may in fact be \emph{unknown}---the problem of computing the overall functionality of a circuit cannot be cast into the sum-over-paths methodology of the continuous-time scenario.  
Instead, one requires a sum-over-paths methodology that is explicitly catered to the circuit scenario, wherein each gate in the circuit is treated as a black box.

It is straightforward to express the transition amplitude of a circuit in terms of a sum or integral over discrete-time paths.  
Suppose $q$ is a label for the basis relative to which we compute amplitudes on a given system---called the {\em configuration} of that system.  For a circuit acting on $n$ systems, the configuration of the $n$ systems is a vector $\vec{q} \equiv (q^{(1)},\dots,q^{(n)})$, where $q^{(i)}$ is the configuration of the $i$th system.  Suppose the circuit is a sequence of $N$ unitaries, $\{ \hat U_k \}_{k=1}^N$,
 so that the total unitary is $\hat U= \hat U_N \hat U_{N-1} \cdots \hat U_2 \hat U_1$.  
It is then appropriate to discretize time into $N$ steps. Denoting the configuration at time step $k$ by $\vec{q}_k \equiv (q_k^{(1)},\dots,q_k^{(n)})$, a discrete-time path through the configuration space is a sequence of $N+1$ configurations,
\be{
\gamma = ( \vec q_0, \vec q_1, \ldots, \vec q_{N}).
}  
Fig.~\ref{gencircuit} depicts a circuit acting on $n$ systems with a gate depth of $N$ and illustrates our labelling convention for the discrete-time paths.

 \begin{figure}[ht]
 \resizebox{9cm}{!}{ 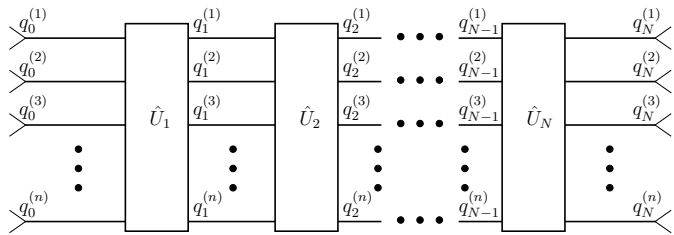 }
 \caption{A circuit consisting of a sequence of $N$ unitaries, together with a parameterisation of the path through configuration space. 
 }
 \label{gencircuit}
 \end{figure}

If the configuration is a continuous variable---for instance, if the Hilbert space of each system is $L^2(\R)$ so that  $q^{(i)}\in \R$---then
we can insert  resolutions of the identity between every pair of adjacent unitaries to obtain
\be{
\langle \vec q_N | \hat U |  \vec q_0 \rangle 
= \int \prod_{k=1}^N \avg{ \vec q_k}{\hat U_k}{\vec q_{k-1}} \dee \vec q_{N-1} \cdots \dee \vec q_1,
\label{jjj}
}
where $\dee \vec q_{k} \equiv \dee q_k^{(1)} \ldots \dee q_k^{(n)}$. 
Defining the amplitude associated with the path $\gamma \in \R ^{n(N+1)}$ as
\begin{equation}
\label{amplitude}
A(\gamma) =\prod_{k=1}^N \avg{\vec q_k}{ \hat{U}_k}{\vec q_{k-1}},
\end{equation}
the amplitude $\langle \vec q_N | \hat{U} | \vec q_0 \rangle $ can be expressed as the following integral over discrete-time paths
\be{
\langle \vec q_N | \hat{U} |  \vec q_0 \rangle  = \int_{\mathbb P_0(\vec q_0,\vec q_N)} A(\gamma) \dee \gamma,
\label{jjjintegral}
}
where $\mathbb P_0(\vec q_0,\vec q_N)$ denotes the space of discrete-time paths that begin at $\vec q_0$ and end at $\vec q_N$ and where $\int_{\P_0(\vec q_0,\vec q_N)} (\cdot) \dee \gamma$ denotes $\int (\cdot) \dee \vec q_{N-1} \cdots \dee \vec q_1$.

If, on the other hand, the configuration is a discrete variable---for instance, if the Hilbert space for each system is $\C^d$ so that our label is discrete, $q^{(i)}\in \Z_d$---then we have
\be{
\langle \vec q_N | \hat{U} |  \vec q_0 \rangle 
= \sum_{\vec q_{N-1} \in (\Z_d)^n} \cdots \sum_{\vec q_{1} \in (\Z_d)^n}  \prod_{k=1}^N \avg{ \vec q_k}{\hat{U}_k}{ \vec q_{k-1}}.
\label{jjj2}
}
Defining the amplitude for a discrete-time path $\gamma \in (\Z_d) ^{n(N+1)}$ by Eq.~\eqref{amplitude}, we have 
\be{
\langle \vec q_N | \hat{U} |  \vec q_0 \rangle  = \sum_{\gamma \in  \P_0(\vec q_0,\vec q_N)} A(\gamma),
\label{jjjsum}
}
where $\P_0(\vec q_0,\vec q_N)$ denotes the space of discrete-time paths that begin at $\vec q_0$ and end at $\vec q_N$.

\black



Under various circumstances, it is possible to restrict the set of paths appearing in the sum or integral.  The paradigm example of this occurs in an interference experiment, where if the particle is known to pass through a screen containing slits, then one can restrict the path integral to those paths that pass through one of the slits.  For instance, if the particle reaches the plane of the screen at the $k$th time step and the slit in the screen is at position $x$, then the transition amplitude for step $k$ has the form $\langle  q_k |\hat U_k | q_{k-1}\rangle \propto  \delta(q_{k} -q_{k-1})\delta(q_{k} -x)$, where $\delta$ represents the Dirac-delta function.  Integrating over $q_{k}$, the delta function forces all paths to pass through the point $q_{k}=x$, so that one can restrict the integral to these paths alone.

Another example of a restriction on the set of paths---the one that will be important here---is when a gate in the circuit maps the set of configurations to itself via some bijective map.  In this case, we have $\langle \vec q_k |\hat U_k |\vec q_{k-1}\rangle \propto  \delta(\vec q_{k} -f(\vec q_{k-1}))$ for some bijective function $f$, and it is sufficient to restrict the sum over paths to those paths for which $\vec q_k = f(\vec q_{k-1})$.

\color{black}

In the continuous-time scenario, one seeks to determine the transition amplitudes for a unitary that is generated by a Hamiltonian (possibly time-dependent) 
over some time interval.  
This is achieved by partitioning the time interval into a large number of small intervals and factorizing the unitary into a sequence of unitaries, one for each time step.  In the limit of small step size, it is well known that the functional over paths appearing in the path sum has the form
\begin{equation}
A(\gamma) = \N e^{i S[\gamma]},
\label{kkk}
\end{equation}
where $ S[\gamma]$ is the classical action\footnote{Here and elsewhere in the article, any action for a path on a continuous-variable configuration space will be assumed to be dimensionless when it appears in a path integral expression, that is, it will be expressed in units of $\hbar$ in that context.} of the path $\gamma$ and $\N$ is a complex number that is independent of the path.
Much of the success of the sum-over-paths methodology relies on the fact that only the {\em phase} and not the {\em magnitude} of the amplitude $A(\gamma)$ is path-dependent, and this in turn is a consequence of taking the limit of small step size.

This fact generally fails to hold in the circuit scenario.  If the gates of the circuit are black boxes, then the
most fine-grained sequence into which the overall unitary associated with the circuit can be factorized is one wherein each element of the sequence corresponds to a gate in the circuit.
But for an arbitrary gate, the associated unitary $ U_k$ 
has matrix elements $\avg{\vec q_k}{U_k}{ \vec q_{k-1}}$ for which the phase {\em and} the magnitude may be dependent on $\vec q_k$ and $\vec q_{k-1}$.  Therefore, in general, both the phase and magnitude of the amplitude $A(\gamma)$ may be dependent on the path $\gamma$\footnote{
 Functionals $A(\gamma)$ wherein both the phase {\em and the magnitude} are path-dependent (corresponding to a complex action functional) have seen applications in the sum-over-paths methodology for continuous-time dynamics (e.g.,\cite{Balian,Behtash}). Whether similar generalizations of the standard sum-over-paths methodology can be of use in the circuit scenario is an interesting question which we do not pursue here.}


Nonetheless,  the functional that appears in the path sum can have  a form analogous to Eq.~\eqref{kkk} {\em for specific types of quantum circuits}.
 This occurs when all of the gates appearing in the circuit have the property of being {\em balanced}.  
 The property of being balanced is defined relative to the orthogonal basis used to define the configuration space in the path sum.  It holds when every transition among basis elements that has a nonzero probability of occurence under the gate has the {\em same} probability of occurence. 
As an example, for a qubit wherein the configuration space in the path sum is the basis of eigenstates of the $Z$ Pauli operator, 
the gate associated to the $Z$ Pauli operator is balanced, as is every gate associated to a linear combination of the $X$ and $Y$ Pauli operators. 
The $Z$ Pauli leaves each basis element invariant, so that only the trivial transition has nonzero probability, while for the $X$ and $Y$ Paulis, each basis element is taken to an equal superposition of basis elements, so that the state transitions that have nonzero probability of occurence is the full set, and each occurs with equal probability. 
 
We now provide a precise definition of the balanced property. 
 Although each unitary in the sequence $\{\hat U_k\}_{k=1}^N$ may in general act on all $n$ systems, it is common to consider gate sets with gates that act on small subsets of the systems  (as will be the case in the Clifford circuits we study further on).  We therefore define the property of being balanced for a gate $\hat{U}$ where the number of inputs and outputs is $m$, which may differ from $n$.
 \black
 \begin{definition}
\label{defnbalancedgate}
Let $|\vec q\rangle $ denote the basis elements for the inputs of the gate and let $|\vec Q\rangle$ denote the basis elements for its outputs.  The gate $\hat{U}$ is said to be balanced if 
\begin{equation}\label{eq:balancedgate}
\forall \vec q,\vec Q\in \R^m \ (or\ (\Z_d)^m): \  \langle \vec Q| \hat{U} | \vec q\rangle = \mathfrak{N} e^{i S( \vec q,\vec Q)}\delta(g(\vec q,\vec Q))
\end{equation}
where $\mathfrak{N}$ is a complex constant\footnote{Because $\mathfrak{N}$ is complex, this decomposition is not unique.  In what follows, we will adopt the convention of excluding any constant terms from $S(\vec Q, \vec q)$ and putting the resulting factor into $\mathfrak{N}$.}, $S(\vec q,\vec Q)$ is a function of $\vec q$ and $\vec Q$ with values in the field $\R$, $g$ is a smooth map\footnote{Due to technical aspects of the theory of distributions, there are restrictions that one should place on this function so that the resulting distribution is well-defined. As these restrictions are satisfied by all of the functions that we explicitly consider in this paper, we shall ignore these conditions and direct the interested reader to Ref.~\cite{Hormander}~(Chapter 6).} $\R^m \times \R^m \to \R^m$  (or $(\Z_d)^m\times(\Z_d)^m \to (\Z_d)^m$ ) and $\delta$ is a Dirac-delta function \footnote{ Here we make the convention that when $g(\vec q, \vec Q)=0$ for all $(\vec q, \vec Q)$, $\delta(g(\vec q, \vec Q))$ is the constant distribution with value $1$. Indeed, one can think of the $\delta$ function as a map which takes functions $g(\vec q, \vec Q)$ to distributions on $(\vec q, \vec Q)$ and so we are free to adopt this convention.} on $\R^m$ (or a Kronecker delta function on $(\Z_d)^m$).  In other words, for the subset of values of $\vec q$ and $\vec Q$ where the amplitude $\langle \vec Q| \hat{U} | \vec q\rangle $ is nonzero---a subset that one can always specify through a condition of the form  $g(\vec q,\vec Q)=0$---this amplitude is equal in magnitude and differs only in phase. 
 \end{definition}

For a circuit composed entirely of balanced gates, the functional over discrete-time paths appearing in the path sum has the form
\begin{equation}
\label{functional4balancedgates}
\gamma \mapsto A(\gamma) = \mathfrak{N} e^{i S(\gamma)} \delta(g(\gamma))
\end{equation}
where $\mathfrak{N}$ is a complex number that is path-independent, $S(\gamma)$ is a real-valued function of  $\gamma$, and $\delta(g(\gamma))$ is a Dirac delta (or Kronecker delta) function that specifies the paths of nonzero amplitude.   \black
Specifically, if the unitary at time step $k$ is made up entirely of balanced gates, 
so that $\avg{\vec q_k}{ \hat{U}_k}{ \vec q_{k-1}}=\mathfrak{N}_k e^{i S_k(\vec q_{k-1}, \vec q_{k})} \delta(g_k(\vec q_{k-1},\vec q_{k}))$, then we have
\begin{equation}
\label{contribution}
\mathfrak{N} \equiv \prod_{k=1}^N \mathfrak{N}_k, \quad  S(\gamma) \equiv \sum_{k=1}^N S_k(\vec q_{k-1}, \vec q_{k}), \quad  \delta(g(\gamma)) \equiv \prod_{k=1}^N \delta(g_k(\vec q_{k-1} , \vec q_{k})).
\end{equation}
Denoting by $\P(\vec{q}_0, \vec{q}_N)$ the space of paths of nonzero amplitude that start at $\vec{q}_0$ and end at $\vec{q}_N$, one then has for continuous variables,
\begin{equation}
\label{PathSumExpression}
 \langle \vec{q}_N | \hat{U} | \vec{q}_0 \rangle = \mathfrak{N} \int_{\P(\vec{q}_0, \vec{q}_N)} e^{i S(\gamma)} \ \dee \gamma,
\end{equation}
where $\int_{\P(\vec{q}_0, \vec{q}_N)} (\cdot) \dee \gamma$ is integration over $\P(\vec{q}_0, \vec{q}_N)$ with respect to the measure induced by $\delta(g(\gamma))$\footnote{This measure is not simply the induced measure on $\P(\vec{q}_0, \vec{q}_N)$ as a subspace of $\P_0(\vec{q}_0, \vec{q}_N) = \R^{n(N+1)}$. It also takes into account the gradients of the functions $g_k$ at those points. For more details, see  Ref.~\cite{Hormander}~(Chapter 6).}. For discrete variables,
\begin{equation}
\label{PathSumExpression2}
 \langle \vec{q}_N | \hat{U} | \vec{q}_0 \rangle =  \mathfrak{N} \sum_{\P(\vec{q}_0, \vec{q}_N)} e^{i S(\gamma)}.
\end{equation}


The notion of balanced gates was introduced in the context of discrete systems by Dawson {\em et al.}\,\cite{Dawson}, who were also the first to consider the sum-over-paths methodology in the circuit scenario\footnote{The definition of a balanced gate used in \cite{Dawson}, that the non-zero matrix elements $\langle Q|\hat{U} |q\rangle$ all have the same {\em absolute value}, is equivalent to the one we have provided.}.  They noted that certain gate sets that are universal for quantum computation---such as the gate set consisting of only the Hadamard and Toffoli gates---are comprised entirely of balanced gates.  As such, circuits built from this gate set can be analyzed by a sum-over-paths methodology, and this was used to provide simple proofs of some known complexity results, for example that $\BQP\subseteq \PP$.
Bacon {\em et al.}\,\cite{AlgCircuits} extended their work by considering {\em algebraic circuits} defined by a gate set consisting of three phase-changing gates and a Fourier transform gate.  Because the elements of this gate set are also all balanced, it is possible to apply the sum-over-paths methodology to algebraic circuits as well. 



Certain well-studied families of circuits, known as Clifford circuits, also have gate sets comprised entirely of balanced gates\footnote{It turns out that for the quopit Clifford circuits and continuous-variable circuits that we consider, balancedness is a property not just of the Clifford gates, but also of any unitary that is implemented by a circuit composed of those gates. For a proof of the balancedness of unitary operations implemented by quopit Clifford gates, see \cite{GK}. The proof for CV Clifford gates is similar.}.
Clifford circuits were first introduced in the context of qubits \cite{Gottesman96}, but were subsequently generalized to continuous variable systems \cite{Barnes} and $d$-level systems for $d>2$ (qudits) \cite{GottesmanCliff, Hostens}.  Dawson {\em et al.}\@ noted that the balanced property held for qubit Clifford circuits.  It is not difficult to see that it holds for continuous variable (CV) and qudit Clifford circuits as well.

In all such circuits---indeed, any circuit consisting entirely of balanced gates---the sum-over-paths methodology provides an alternative way of computing transition amplitudes for the whole circuit from a knowledge of the transition amplitudes of each gate.

In this article, we are {\em not}, however, interested in the sum-over-paths approach for its use as an alternative technique of solving quantum dynamics for circuits, but rather for the novel perspective that it offers on the difference between quantum and classical theories of that circuit. 

In the case of continuous-time dynamics, the bridge between the classical and the quantum theory is made through the phase factor $S[\gamma]$ that is assigned to a path $\gamma$ in the path integral expression for the dynamics; it is simply the classical action of the path $\gamma$.  Specifically, in the case of $n$ systems described by continuous variables undergoing continuous-time dynamics over a time interval $[0,T]$,  a path $\gamma$ is specified, in the limit of small step-size, as a function $\vec q: [0,T] \to \R^n$,
and  the classical action of this path is the integral of the Lagrangian of the system along the path,
\begin{equation}
S[\gamma]\equiv \int_0^T  \mathcal{L} (\vec q(t), \dot{\bt q} (t) )\;\dee t.
\label{actionfunctionaldefn}
\end{equation}

We here address the following question:   for {\em discrete-time dynamics},
can the phase $S[\gamma]$ appearing in the path integral expression 
be understood as the {\em action functional} for a discrete-time path $\gamma$ in a classical counterpart of the quantum circuit?  The main result of our article is a demonstration that it {\em can} be so understood for certain kinds of circuits.  

\subsection{Summary of results}

Our demonstration that the phase of a discrete-time path can be understood in terms of an action functional necessitates solving two problems about discrete-time classical dynamics:
\begin{enumerate}
\item Defining an action functional for discrete-time classical dynamics, \label{step1}
\item Determining the discrete-time {\em classical} dynamics associated to a given type of discrete-time {\em quantum } dynamics. 
\end{enumerate}
Our solutions to these two problems will be outlined in the rest of the introduction and elaborated upon in the main text.  Once these two foundations are laid, it becomes straightforward to verify that for certain types of quantum circuits, the action functional for their classical counterparts is what determines the relative phases of paths in the sum-over-paths expression for the quantum dynamics.  

{\bf Defining an action functional for discrete-time dynamics.}
To understand what counts as a good definition of an action functional in discrete-time dynamics, we review the role that the action functional plays in continuous-time dynamics.  In fact, it plays two related roles:
\begin{itemize}
\item It determines the classical trajectories via a least-action principle; 
\item It generates the symplectomorphism associated with evolution over a time interval via its evaluation on the classical trajectories over that interval;
\end{itemize}
The first role is well-known, while the second role is worth reviewing.
Suppose that one is considering continuous-time dynamics defined by a Hamiltonian $\tt H(\bt q,\bt p)$ over a time interval $[0,T]$  and let $(\bt q_\mathrm{cl}(t),\bt p_\mathrm{cl}(t))$ be the solution to Hamilton's equations with initial value $(\bt q, \bt p)$. Then the dynamics over the given time interval is described by the symplectomorphism $\phi:(\bt q,\bt p) \mapsto (\bt Q,\bt P)$ by setting $\bt Q = \bt q_\mathrm{cl}(T)$ and $\bt P = \bt p_\mathrm{cl}(T)$. 
 Note that $\bt q_\mathrm{cl}(t)$ can also be characterized as the solution to the Euler-Lagrange equations with boundary conditions $\bt q_\mathrm{cl}(0)= \bt q$ and $\bt q_\mathrm{cl}(T)=\bt Q$ for the Lagrangian ${\cal L}(\bt q, \dot{\bt q})$ associated to $\tt H(\bt q, \bt p)$.
 Consider the function of $\bt q$ and $\bt Q$ that one obtains by evaluating the action functional on the classical trajectory through configuration space that begins at $\bt q$ and ends at $\bt Q$, denoted $\gamma_{\rm cl}$, 
\begin{equation}\label{generatingfunction}
\tt G_{\phi}(\bt q, \bt Q) \equiv S[\gamma_{\rm cl}] = \int_{0}^{T} {\cal L}(\bt q_{\mathrm{cl}}(t), \dot{\bt q}_\mathrm{cl}(t)) \ \dee t.
\end{equation}
 It is well-known (\cite{Arnold}, Chapter 9) that this function generates the symplectomorphism $\phi$ in the sense that
\begin{equation}
\label{genfunchar}
p^{(i)} = - \frac{\partial G_{\phi}(\vec q,\vec Q)}{\partial q^{(i)}} , \ P^{(i)} = \frac{\partial G_{\phi}(\vec q,\vec Q)}{\partial Q^{(i)}}.
\end{equation}
In other words, the generating function of a symplectomorphism induced by Hamiltonian dynamics is exactly the action functional evaluated on the classical trajectories. \color{black}

Now consider {\rm discrete-time} classical dynamics for a continuous-variable system, such as the classical interferometer described earlier.  Suppose that the nature of the dynamics is specified modularly, as a sequence of symplectomorphisms on the system's phase space. 
For both the case of a physical theory wherein time is fundamentally discrete and the case of circuit dynamics wherein the internal dynamics of each gate is unknown, there is no Lagrangian describing a continuous-time dynamics within a given time-step.  Nonetheless, the association between generating functions and actions described above provides a means of defining an action without reference to such a Lagrangian.  Specifically, if the symplectomorphism for the $k$th time-step is denoted $\phi_k$ and has a generating function $G_{\phi_k}(\bt q_{k-1}, \bt q_k)$, one can simply interpret the latter as the total action of the transition from $\bt q_{k-1}$ to $\bt q_k$ under $\phi_k$.  One can then define the action functional for the discrete-time path $\gamma = (\bt q_0,\ldots, \bt q_N)$ as the sum over time-steps of the actions associated to each transition,
\begin{equation}\label{CVactionfunctional}
S(\gamma) = S(\bt q_0,\ldots, \bt q_N)  = \sum_{k=1}^N \tt G_{\phi_k}(\bt q_{k-1}, \bt q_k).
\end{equation}
The reason it is appropriate to identify this as the discrete-time action functional is that, as we will prove in Proposition~\ref{actclasstraj}, it plays one of the roles that the continuous-time action functional plays, namely, that if one evaluates it on the classical discrete-time trajectories, one obtains the generating function for the symplectomorphism associated to the overall transition between the initial time and the final time. 

We will show that this result about discrete-time dynamics goes through not just for continuous variables but for discrete variables as well.  In the latter case, we imagine that a single system has a phase space with coordinates taking values in the ring $\Z_d$, so that the phase space of $n$ such systems is $(\Z_d)^{2n}$, and the discrete-time dynamics is specified as a sequence of symplectomorphisms on this phase space.  At first glance, it is not obvious that the connection between generating functions and actions that was described above can be leveraged to define a discrete-time action functional in this case.
The difficulty is that, as is evident from Eq.~\eqref{genfunchar}, generating functions are typically characterised in a way that makes explicit use of differential calculus, and for discrete variables, one does not have a notion of differentiation in the usual sense. 


To resolve this issue, we call upon the {\em algebra-geometry correspondence}. While this has many incarnations in various areas of mathematics, the underlying idea is that one can frequently establish a dictionary which translates geometric structures of a space into structures on its algebra of functions, and vice-versa. The particular instantiation we need is the duality between the geometry of so-called affine schemes and the algebra of commutative rings (\cite{Hartshorne}, Chapter II.2). 

Under this duality, the algebraic counterpart of an affine space over $\Z_d$ is an algebra of polynomials with coefficients in $\Z_d$, where the number of variables equals the dimension of the affine space. Explicitly, for the classical phase space of $n$ systems with coordinates $q^{(i)}, p^{(i)}$, each taking values in $\Z_d$, one has the correspondence
\begin{equation}
(\Z_d)^{2n} \leftrightsquigarrow \Z_d[q^{(1)},\ldots, q^{(n)}, p^{(1)}, \ldots, p^{(n)}] =: \Z_d[\vec q, \vec p].
\end{equation}
We exploit this correspondence by identifying what structure on the algebra of polynomials is the dual of a differential structure on the discrete space. Fortunately this question has long been answered by algebraic geometers. The structure goes by the name of the K\"ahler differential forms on the algebra $\Z_d[\vec q, \vec p]$ (\cite{Hartshorne}, Chapter II.8). 

 The algebra-geometry correspondence {\em forces} us to define generating functions for the dynamics of discrete systems not as polynomial functions over the reals,
 but rather as polynomial functions over $\Z_d$,
\begin{equation}
G(\vec{g}, \vec{Q}) \in \Z_d[\vec q, \vec Q].
\end{equation}
For example, for any symplectomorphism $\phi$ associated to a single system, that is, $\phi:\Z_d \times \Z_d \to \Z_d \times \Z_d$ and $( Q,P)=\phi[(q,p)]$, the generating function is a polynomial in two variables, $q$ and $Q$, with coefficients in $\Z_d$, 
\begin{equation}
G_{\phi}(q,Q): \Z_d \times \Z_d \to \Z_d.
\end{equation}
Similarly for symplectomorphisms on pairs of systems, $\phi:(\Z_d)^2 \times (\Z_d)^2 \to (\Z_d)^2 \times (\Z_d)^2$, the associated generating function is a polynomial with coefficients in $\Z_d$,
\begin{equation}
G_{\phi}(\vec q,\vec Q): (\Z_d)^2 \times (\Z_d)^2 \to \Z_d.
\end{equation}

It follows that if we define the total action of a given transition under a symplectomorphism by evaluating the associated generating function on the given initial and final configurations, then the action will be $\Z_d$-valued.  We can then define a $\Z_d$-valued action functional for a discrete-time path as the sum over time-steps of the action of the transition at each time-step, so that Eq.~\eqref{CVactionfunctional} applies, but is $\Z_d$-valued.  Just as in the continuous-variable case, if one evaluates this discrete-time action functional on the classical discrete-time trajectories, one obtains the generating function for the symplectomorphism associated to the overall dynamics of the whole circuit, as we show in Lemma~\ref{actclasstraj_discrete}.  It is in this sense that it is appropriate to understand the sum of the generating functions as an action functional, even for discrete variables.\footnote{Note that our proposal agrees with that of Baez and Gilliam \cite{Baez} concerning what sorts of mathematical object should represent a discrete-time and discrete-variable analogue of the action functional. }

We have shown that our discrete-time action functional plays the second role that its continuous-time counterpart plays, but what about the first role?  Does it also define classical discrete-time trajectories by a least-action principle?


A bit of reflection shows that the question is not well-posed without substantial revision to the statement of the principle. For discrete degrees of freedom in particular, the normal manner of thinking about it---in terms of an {\em extremization} of the action functional---does not work.  As we noted above, the algebra-geometry correspondence forces the action to be $\Z_d$-valued in this case, and $\Z_d$ is not an ordered set, so that there is no possibility of defining a principle of extremality in terms of it.  It may be possible, however, to recast the principle in terms of stationary points in the space of paths rather than in terms of extremization. To do so, one must translate to the discrete-variable context the normal story about how to express the variation of the action induced by a small variation in the path.  But it is not even clear what is the appropriate sort of variation of the path to consider.  
Presumably, if this question can be answered, then stationary points can be identified using a scheme wherein Kahler differential forms play the role of normal derivatives. 
In the case of continuous degrees of freedom, where the action functional is real-valued, the question is well-posed, but we do not have the answer.  Whether our proposal for a discrete-time action functional can recover classical paths by some principle of stationary action (for either discrete or continuous degrees of freedom) remains an open question.

Because it remains unclear whether it is possible to understand discrete-time classical dynamics in terms of a principle of stationary action, 
in this work we focus on demonstrating that our proposed functional over discrete-time paths plays the {\em second role} that the action functional over continous-time paths plays in classical dynamics, that of defining a generating function.  It is in this sense that our proposed functional does indeed merit the title of an action functional.

\color{black}

{\bf Determining the classical counterpart of quantum discrete-time dynamics.} For certain kinds of quantum circuits involving only balanced gates, we can identify the natural classical counterpart and its description in terms of a phase space and symplectomorphisms.  In particular, we do so for CV Clifford circuits and for qudit Clifford circuits where the dimension $d$ of the elementary systems is an odd prime. 



The CV case is the easiest to consider because there it is obvious which classical phase space to associate to a given quantum state space and
(at least for the CV Clifford gates) 
which symplectomorphisms to associate to given quantum unitaries.   
This is achieved using the Wigner representation, which defines a classical model of the quantum dynamics (specifically, a noncontextual hidden variable model) ~\cite{Bartlett2012}.
The Wigner representation associates a symplectomorphism to each of the elementary gates of a CV Clifford circuit, and from the latter we obtain a generating function.  In this way, we can identify the action functional for the overall circuit.  We then proceed to show (in Theorem~\ref{MainCV}) that this action functional yields the correct phases in the path sum expression for the quantum dynamics, that is, we show that it yields the functional $S(\gamma)$ of Eq.~\eqref{functional4balancedgates}.
 







To achieve an analogous interpretation of $S(\gamma)$ in Eq.~\eqref{functional4balancedgates} for discrete rather than continuous-variable systems in a circuit scenario, one must first of all determine what  symplectic space should be associated with a given discrete quantum system.  
This is not evident a priori because for discrete systems, such as the intrinsic spin degree of freedom, the quantum dynamics was not obtained historically by quantizing a classical theory of discrete variables.

However, it turns out that for certain qudit Clifford circuits, there is clarity about what is the natural classical counterpart.  These are the so-called  ``quopit'' Clifford circuits, where a quopit is a qudit where the dimension $d$ is an odd prime~\cite{Emerson}.  What is special about quopit Clifford circuits is that they are known to admit a noncontextual hidden variable model, and this model provides the classical counterpart of the quantum dynamics\footnote{This is not at odds with the fact that the full quantum theory fails to admit of such a model because Clifford circuits do not realize arbitrary unitaries.}.
For quopit Clifford circuits, the model is provided by a discrete analogue of the Wigner representation proposed by Gross, where the Clifford operations are represented as transformations of an affine space over the finite field $\Z_d$~\cite{Gross}.\footnote{In fact, Gross's Wigner representation is defined for all odd dimensions, not just odd prime. In that case $\Z_d$ is no longer a field, which significantly increases the difficulty of working with this representation. \black Nonetheless, the elementary gate set for composite dimensions has been worked out by Hostens \cite{Hostens}, but it is much more complicated and we do not consider it here. We believe that similar results as ours should hold for arbitrary odd dimensions.
}


The discrete Wigner representation associates to every elementary gate of a quopit Clifford circuit a symplectomorphism on the discrete phase space.\footnote{The symplectomorphisms arising from both the discrete and continuous variable Wigner representations are all {\em affine}, that is, a composition of a linear symplectic map and a phase-space translation. We nonetheless propose, in Sections \ref{genfunctions} and \ref{genfundisc}, the discrete-time analogues of action functionals for dynamics given by arbitrary (i.e., not necessarily affine) symplectomorphisms. As the formalism employed naturally accommodates this broader class of dynamics it adds no further mathematical technicalities while greatly expanding the range of applicability of our proposal.} 
Using the techniques described earlier, one can define generating functions in terms of these and consequently also an action functional. 
We show (in Theorem~\ref{MainDisc}) that this choice does indeed yield the correct expressions for the quantum dynamics, that is, we show that the functional $S(\gamma)$ of Eq.~\eqref{functional4balancedgates} can indeed be written as a sum of the generating functions of the symplectomorphisms associated to the gates. 


\black

The idea of looking at sums of generating functions as generalisations of the action functional is in part inspired by a little-known paper of Dirac \cite{Dirac} wherein he explores the possibility of a Lagrangian approach to quantum mechanics.  While Dirac was not successful at reformulating quantum mechanics, he did notice certain {\em formal similarities} between the generating functions of symplectomorphisms and the infinitesimal generators of unitary operations that implement a change of basis, and it was this work that ultimately inspired Feynman's formulation of the path integral (as Feynman notes in his Nobel lecture).  By exploring this connection in the discrete-time scenario, our results serve to clarify the precise role of generating functions in the sum-over-paths formulation of quantum theory.
\newline

\subsection{Significance for understanding the quantum-classical distinction}

 The physical relevance of whether or not the exponent in the path integral expression can be interepreted as a classical action functional is this: if it can be so interpreted, then one has built a bridge between a classical theory and its corresponding quantum theory for discrete-time dynamics.  
In particular, one obtains insights into two aspects of this bridge:
\begin{enumerate}
\item Schemes for quantization, that is, how to define the quantum counterpart of a given classical evolution,
\item The definition of intrinsically quantum behaviour, that is, determining when a given experiment fails to admit of a classical explanation.
\end{enumerate}
Although the first aspect concerns a map from classical to quantum, while the second concerns a map from quantum to classical (or the lack thereof), there remains an important distinction between them.
For the problem of quantization, one has a particular classical phase-space and Hamiltonian in mind. For the problem of determing whether an experiment admits of a classical explanation, on the other hand, one would like to be permissive about the nature of the classical phase-space and Hamiltonian that can appear in the explanation.

We discuss each aspect in turn.  

{\bf Quantization.} For continuous-time dynamics,  the sum-over-paths approach provides a means of making inferences {\em from} the Lagrangian description of the classical dynamics {\em to} its quantum dynamics. Indeed, it is this sort of problem that has driven the development of the vast technical machinery and wide-ranging applications of the path integral.  The sum-over-paths methodology arguably provides the most fruitful approach to the problem of quantization.  In particular, recall that, unlike canonical quantization, the path integral formulation of quantum theory can be applied to classical theories that have a Lagrangian formulation but no Hamiltonian formulation.  Consequently, the path integral approach has a broader scope.


In terms of applications, understanding the sum-over-paths methodology for the circuit scenario may provide a means to directly ``quantize'' certain classical circuits.  Studying the quantum generalization of various types of circuits, codes, and other tools from classical computer science has been a source of many innovations in the field of quantum computation.  The present work deviates from the traditional approach to such generalizations insofar as it requires a phase-space description of the classical circuit.  We hope that the novelty of this perspective may offer new insights. 
 

A sum-over-paths methodology for the circuit scenario can also address the problem of quantizing theories of physics wherein time is fundamentally discrete.  
Many have espoused the idea that discrete-time dynamics might be the correct basis for physics while the standard continuous-time dynamics might be merely a useful approximation thereto.  In the classical context, the idea has been pursued through the study of cellular automata.  The fact that cellular automata have a dynamical law that is similar to laws of physics in being both time-independent and spatially local, and the fact that many choices of this law yield dynamics having features that are strongly reminiscent of physics, including
a fundamental limit to the speed of propagation of influences, the possibility of evolving stable structures and complex structures, and computational universality, 
has motivated many to pursue a reconstruction of physical theories in terms of them~\cite{fredkin1990informational,Toffoli, Wolfram}.
 In the quantum context, the fact that it is in principle impossible to resolve spatial distances and times arbitrarily finely~\cite{ng1995limitation,amelino1994limits} and the idea that lengths and times, like other observables in quantum theory, ought to take values in a discrete spectrum~\cite{rovelli1995discreteness,Wallden} also motivate researchers to pursue formulations of physics wherein space and time are fundamentally discrete, implying discrete-time rather than continuous-time dynamics.  Some have also considered quantum cellular automata (see, e.g., \cite{watrous1995one,schumacher2004reversible}) as a basis for physics~\cite{d2017quantum}. 

Generally, researchers pursuing discrete-time dynamics as the basis of physics favour the assumption that the internal degrees of freedom of systems should be taken to be discrete.  Nonetheless, a discrete-time dynamics for {\em continuous} degrees of freedom is another consistent option. 
For instance, in the quantum context, one can consider scalar fields as the systems which evolve over discrete time,
and in the classical context, one can consider cellular automata where the internal state of each cell is a continuous variable \cite{Wolfram}. 

It is our hope that the results in this article might provide some insight into how to achieve quantization in such discrete-time classical theories. 
 
 For instance, if one constructs a cellular automaton wherein the state-space of a cell can be understood as a phase space and the update rule (the discrete-time dynamics) can be understood as a symplectomorphism, then our results provide a way of determining the quantized version of that cellular automaton.  In this sense, our work connects most with an approach to classical discrete-time dynamics termed ``discrete mechanical systems'', where, unlike standard cellular automata, one aims for a discrete generalization of Hamiltonian and Lagrangian descriptions of the dynamics~\cite{Baez}.  

 Note that if one starts from classical discrete-time theories wherein the dynamics is generated by a set of  symplectomorphisms, and one uses the sum-over-paths methodology outlined here to determine the quantum counterpart of this theory, then the discrete-time quantum theories that one obtains can always be formulated as a kind of quantum circuit.  Importantly, however, not every sort of quantum circuit that the Hilbert-space formalism of quantum theory permits us to define can arise in this fashion.  This method of quantization can only yield quantum circuits wherein each of the fundamental gates has the property of being balanced.  Our work therefore provides some reason for thinking that this restricted class of quantum circuits may be a better starting point for any research program that aims to build a quantum formulation of physics wherein space and time are fundamentally discrete. \color{black}

{\bf Intrinsically quantum behaviour.}   
Discussions of the distinction between quantum and classical dynamics often appeal to the path integral methodology: a given quantum dynamics is thought to admit of a classical model (hence to not be intrinsically quantum) if the typical action scale is large compared to Planck's constant, such that the amplitudes of paths which deviate from the action-extremizing path tend to cancel. It is not at all clear, however, how this notion of quantum-classical correspondence might be applied to a {\em discrete-time} evolution, such as arises in a quantum circuit. 
 
Meanwhile, the notions of {\em local causality}~\cite{Bel64a} and of {\em noncontextuality}~\cite{Bel66a,KS} are naturally suited to the question of whether a given quantum circuit admits of a classical explanation.
 In this approach, one considers the set of possible circuits that can be built up from the elementary gate set, and one inquires about the possibility of explaining the experimental data for each of these in terms of a locally causal model or in terms of a noncontextual model.  The notion of a noncontextual model of an experiment was generalized in Ref.~\cite{Spekkens05} and the generalization was shown to be equivalent to the existence of a nonnegative quasiprobability representation~\cite{Spekkens08,FerrieEmerson}, another popular notion of classicality.  
 Note that the notion of classicality that emerges in such work  does {\em not} involve the ratio of typical actions to Planck's constant. 

It is our view that we will not have understood the quantum-classical distinction until such time as we have a notion of intrinsically quantum behaviour that is independent of how one formulates quantum theory (whether it be with path integrals or Schr\"{o}dinger dynamics, for instance) and of whether one is considering continuous-time or discrete-time dynamics.  Insofar as our work contributes to understanding the quantum-classical distinction in the path integral formulation of discrete-time dynamics, it is a step on the path towards such a unified notion of intrinsic quantumness.


The idea that a quantum dynamics should be deemed effectively classical only if the associated sum over paths is dominated by a single path (the action-extremizing one) 
has recently been challenged by Kent~\cite{Kent}.  Our work provides further reason to be sceptical of this notion of classicality, independently of the reasons provided by Kent.  If one understands a quantum dynamics to admit of a classical explanation when it admits of a noncontextual model (or, equivalently, a nonnegative Wigner representation), then Clifford circuit quantum dynamics admit of a classical explanation.  And yet, as we will show here,  
the path-sum expression for such dynamics cannot be reduced to the contribution of a single path. 
Our work therefore provides a starting point for a re-assessment of what is the correct notion of the classical limit in the path integral formulation of quantum theory.

Finally, the question of which families of quantum circuits admit of a classical model has practical significance: it can help to identify the resources that are responsible for quantum-over-classical advantages for information processing. 
It has recently been shown that certain types of nonclassicality 
can constitute resources for cryptographic and computational tasks.  For instance, Bell-inequality violations are a resource for device-independent cryptography\cite{Barrett, Pironio}.  Furthermore, failing to admit of a noncontextual model (equivalently, failing to admit of a positive quasiprobability representation) has recently been implicated in quantum computational speed-up \cite{Anders, Raussendorf, HowardEtAl}.  
A broader perspective on nonclassicality in the circuit scenario promises more such insights. It may also help to determine whether a given computational architecture, such as the one implemented by D-wave, has intrinsically quantum features or not\cite{Albash}.   

\color{black}

\vspace{.5cm}
{\bf Outline.}
The outline of the paper is as follows.
The continuous variable and quopit Clifford circuits are considered separately in Sections \ref{CVClifford} and \ref{Clifford}, respectively.  
This is because, although the end results are quite similar, the mathematics involved is quite different.  After introducing these circuits, in Sections \ref{CVPathSum} and \ref{PathSum}, we explicitly describe the resulting sum-over-paths expressions for transition amplitudes in Theorems \ref{CVClPathSumFinal} and \ref{penult}. The remainder of each Section is devoted to showing that the functionals $S(\gamma)$ (of \eq{functional4balancedgates}) are the generating functionals of the corresponding symplectic representations. More specifically, for the continuous variable case, we introduce generating functions in Section \ref{genfunctions} and then, in Section \ref{CVGenAction}, we introduce the symplectic representation and prove Theorem \ref{MainCV}. For the quopit case, we introduce the symplectic representation via the discrete Wigner transform in Section \ref{Symprep}. Then, in Section \ref{genfundisc}, we introduce K\"ahler differentials and use them to define the corresponding generating functional. Finally, in Section \ref{DiscGenAction}, we prove Theorem \ref{MainDisc}. Section \ref{Conclusion} offers some concluding remarks.  As some of the mathematics used in Section \ref{Clifford} may not be familiar to some readers, we have included an Appendix  with some additional background information.

\vspace{.5cm}
\noindent \textbf{Notational Conventions.} 
Arrows ($\vec{x}$) indicate vector quantities and bracketed superscripts ($x^{(i)}$) their components. Hats ($\hat{x}$) indicate operators. Subscripts ($x_k$) will be used to index time steps.  $\Z_d$ denotes the ring of integers modulo $d$.  Complex conjugation will be represented by an overbar ($\bar{z}$).

\section{Continuous variable Clifford circuits} 
\label{CVClifford}
\subsection{Sum-over-paths expression for CV Clifford circuits}
\label{CVPathSum}
We turn our attention to applying the sum-over-paths methodology to a particular example of a family of quantum circuits for which every gate in the generating set is balanced: the subset of quantum circuits known as {\em continuous variable (CV)  Clifford} circuits. These have previously been studied as the appropriate generalization of qubit Clifford circuits for continuous variables \cite{Barnes}.
 In fact, it has been shown that such circuits can be efficiently simulated on a classical computer, an extension of the Gottesman-Knill Theorem from qubits to CV systems~\cite{Bartlett}. Our interest in CV Clifford circuits comes from a more foundational perspective, namely, that they can be described by a noncontextual hidden variable model \cite{Bartlett2012}, which provides the means by which we identify an action functional over the paths, as we shall see in Sec.~\ref{CVGenAction}.

The goal of this section is to determine, for an arbitrary CV Clifford circuit, a sum-over-paths expression for its transition amplitudes as in Eq.~\ref{PathSumExpression}.
We then show that the exponent of the phase factor associated to each allowed path can be understood as a discrete-time generalisation of the action functional. 
We introduce this notion in Section \ref{genfunctions} and then,  in Section \ref{CVGenAction}, we prove that it agrees with our calculation of the aforementioned phase factor.

An $n$-system CV Clifford circuit consists of preparations and measurements in the configuration basis of $L^2(\R^n)$ and an elementary gate set consisting of the following $1$-system and $2$-system gates:
\be{\label{CVGateSet}
\htt F &=& e^{- i \frac{\pi}{4} (\htt q^2 + \htt p^2) } \nn
\htt P(\eta) &=& e^{-i \frac{\eta}{2} \htt q^2}, \quad \eta \in \R \nn
\htt X(\tau)&=& e^{-i \tau \htt p}, \quad \tau \in \R \nn
\SUM &=& e^{-i \htt q^{(1)} \otimes \htt p^{(2)}},
}
along with $\htt F^\dag$ and $\SUM^\dag$ ($\htt P^{\dag}(\eta)$ and $\htt X^{\dag}(\tau)$ are already included as $\htt P(-\eta)$ and $\htt X(-\tau)$). 
$\htt F$ is called the Fourier gate and corresponds to evolution for unit duration under the Hamiltonian for a harmonic oscillator with mass $\frac{2}{\pi}$ and frequency $\frac{\pi}{2}$ (it is the analogue of the Hadamard gate in a qubit Clifford circuit).  It is intuitively understood as a rotation in phase space by $\pi/2$.  $\htt P(\eta)$ is called a phase gate (by analogy to the phase gate in a qubit Clifford circuit), and corresponds to a phase-space squeezing operation (via a position-dependent boost).  The $\htt X(\tau)$ gate (a generalization of the Pauli-$X$ gate in a qubit Clifford circuit) implements a translation of the configuration by $\tau$. Finally, $\SUM$ is called the sum gate (it is the analogue of the CNOT gate in a qubit Clifford circuit) and can be understood as a translation of the second system by an amount equal to the coordinate of the first.  The intuitive phase-space accounts  that we have just provided for these quantum gates will be shown,  in Sec.~\ref{CVGenAction},  to be an accurate description of the associated symplectomorphisms.

Bartlett {\em et al.} \cite{Bartlett} have shown that CV Clifford circuits can (up to a global phase factor) implement all and only those unitaries lying in the so-called $n$-system CV Clifford group, $\tt C_n$. To define $\tt C_n$, one must first introduce the $n$-system CV Pauli group ${\cal G}_n$. 
Denote the group of unitaries on $L^2(\R^n)$ by $U\left(L^2(\R^n)\right)$.  ${\cal G}_n$ is the subgroup of  $U\left(L^2(\R^n)\right)$ that is generated by 
$\{ \htt X_i(\tau), \htt Z_i(\sigma) : \tau,\sigma \in \R, i\in \{1,\dots,n\}\}$ 
where $\htt X_i(\tau)$ is the operator that translates system $i$ by $\tau$ 
and $\htt Z_i(\sigma) := e^{i\sigma \htt q_i}$ is the operator that boosts system $i$ by $\sigma$.
\begin{definition}
 The $n$-system CV Clifford group, $\tt C_n$, is defined to be the normaliser of the $n$-system CV Pauli group inside $U\left(L^2(\R^n)\right)$, that is, $\tt C_n := N\left({\cal G}_n\right)$.
\end{definition}

Note that $\tt C_n$ has $U(1)$ as a subgroup given by the operators $e^{i \phi} {\mathds 1}$. Then Bartlett {\em et al.} \cite{Bartlett} prove that the set
\begin{equation}
 \left\{ \htt F_i,\htt P(\eta)_i, \htt X(\tau)_i, \Sigma_{i,j} \ : \ \eta, \tau \in \R, i, j \in \{1,\ldots, n\} \right\}
\end{equation}
are a generating set for the group $\tt C_n / U(1)$.

Consider a given $n$-system CV Clifford circuit $\mf C$ implementing a unitary $\htt U \in \tt C_n$. To calculate the amplitudes for each path through the configuration space, we first need the matrix elements for the elementary gates.
\begin{lemma}
\label{CVMatElems}
 The matrix elements for the elementary CV Clifford gates are:
 \begin{itemize}
  \item $\langle \tt Q | \htt F | \tt q \rangle = \frac{1-i}{2 \sqrt{\pi}} e^{-i \tt Q \tt q}$;
    \item $\langle \tt Q | \htt F^\dag | \tt q \rangle = \frac{1+i}{2 \sqrt{\pi}} e^{i \tt Q \tt q}$;
  \item $\langle \tt Q | \htt P(\eta) | \tt q \rangle = e^{- i \frac{\eta}{2} \tt q^2} \delta\left(\tt Q - \tt q\right)$;
  \item $\langle \tt Q | \htt X(\tau) | \tt q \rangle = \delta\left(\tt Q - (\tt q + \tau)\right)$;
  \item $\langle \tt Q^{(1)}, \tt Q^{(2)} | \SUM |\tt q^{(1)}, \tt q^{(2)} \rangle = \delta\left(\tt Q^{(1)} - \tt q^{(1)}\right) \delta\left( \tt Q^{(2)} - (\tt q^{(2)} + \tt q^{(1)})\right)$;
  \item $\langle \tt Q^{(1)}, \tt Q^{(2)} | \SUM^\dag | \tt q^{(1)}, \tt q^{(2)} \rangle =\delta\left(\tt Q^{(1)} - \tt q^{(1)}\right) \delta\left( \tt Q^{(2)} - (\tt q^{(2)} - \tt q^{(1)})\right)$.
 \end{itemize}
 It follows that all of these gates are balanced.
\end{lemma}
\begin{proof}
\begin{itemize}
 \item We use the fact that $\htt F$ corresponds to evolution for unit duration under the Hamiltonian for a harmonic oscillator with mass $\frac{2}{\pi}$ and frequency $\frac{\pi}{2}$, the matrix elements of which are well-known (e.g. problem 3-8 in \cite{FeynmanHibbs}), to infer that
\be{\langle \tt Q | \htt F | \tt q \rangle = \frac{1-i}{2 \sqrt{\pi}} e^{-i \tt Q \tt q}. }
It follows that
\be { \langle \tt Q | \htt F^\dag | \tt q \rangle &=& \overline{\langle \tt q | \htt F | \tt Q \rangle} \nn
 &=& \frac{1+i}{2 \sqrt{\pi}} e^{i \tt Q \tt q}.} 
 \item The matrix elements of $\htt P(\eta)$ are trivial to compute:
\be { \langle \tt Q | \htt P(\eta) | \tt q \rangle = e^{- i \frac{\eta}{2} \tt q^2} \delta\left(\tt Q - \tt q\right). }
\item For $\htt X(\tau)$, one has
\be { \langle \tt Q | \htt X(\tau) | \tt q \rangle &=& \int \dee \tt p  \ \langle \tt Q | e^{-i \tau \htt p} | \tt p \rangle \langle \tt p | \tt q \rangle \nn
&=& \frac{1}{2 \pi} \int \dee \tt p \ e^{i \tt p ( \tt Q - \tt q -\tau)} \nn
 &=& \delta\left(\tt Q - (\tt q + \tau)\right). }
 \item Finally, for $\SUM$ one has
\be{\langle \tt Q^{(1)}, \tt Q^{(2)} | \SUM |\tt q^{(1)}, \tt q^{(2)} \rangle &=& \int \dee \tt p \ \langle \tt Q^{(1)}, \tt Q^{(2)} | \SUM | \tt q^{(1)}, \tt p \rangle \langle \tt p | \tt q^{(2)} \rangle \nn
&=& \frac{1}{\sqrt{2 \pi}} \int \dee \tt p \ e^{i \tt p \left( - \tt q^{(2)} - \tt q^{(1)}\right)} \langle \tt Q^{(1)}, \tt Q^{(2)} | \tt q^{(1)}, \tt p \rangle \nn
&=& \frac{1}{2 \pi} \delta\left(\tt Q^{(1)} - \tt q^{(1)} \right) \int \dee \tt p \ e^{i \tt p \left( - \tt q^{(2)} - \tt q^{(1)} + \tt Q^{(2)}\right)} \nn
&=& \delta\left(\tt Q^{(1)} - \tt q^{(1)}\right) \delta\left( \tt Q^{(2)} - (\tt q^{(2)} + \tt q^{(1)})\right). }
It follows that
\be{ \langle \tt Q^{(1)}, \tt Q^{(2)} | \SUM^\dag | \tt q^{(1)}, \tt q^{(2)} \rangle &=& \overline{\langle \tt q^{(1)}, \tt q^{(2)} | \SUM | \tt Q^{(1)}, \tt Q^{(2)} \rangle} \nn
&=&  \delta\left(\tt Q^{(1)} - \tt q^{(1)}\right) \delta\left( \tt Q^{(2)} - (\tt q^{(2)} - \tt q^{(1)})\right).}
\end{itemize}
\end{proof}
For any system at any time-step where the circuit has no gate acting, we shall describe the gate as {\em identity} and denote it by $\mathds{1}$.  The identity gate is a special case of $\htt X(\tau)$ where $\tau=0$ and a special case of $\htt P(\eta)$ where $\eta =0$, so that we can infer from Lemma \ref{CVMatElems} that its contribution to the amplitude is simply $\delta\left(\tt Q - \tt q\right)$.  (Note that although we could have simply represented the identity gate by $\htt X(0)$ or $\htt P(0)$ in a description of the circuit, it is more straightforward to treat it distinctly).

For a CV Clifford circuit $\mf C$, the amplitude of any path $\gamma$ through configuration space is given, according to \eq{functional4balancedgates},  by 
$$A_{\mf C}(\gamma) = \mathfrak{N}_{\mf C} e^{i S_{\mf C}(\gamma)} \delta(g_{\mf C}(\gamma))$$
 where $\mathfrak{N}_{\mf C}$, $S_{\mf C}(\gamma) $ and $\delta(g_{\mf C}(\gamma))$ can be decomposed into contributions from each time-step in the manner described by \eq{contribution}.  It should be noted that gates acting on different systems during the same time-step contribute to the overall amplitude exactly as they would if they acted at consecutive time-steps, and hence their contributions also combine in the fashion described by \eq{contribution}.  
 Lemma \ref{CVMatElems} then allows us to express the contribution of each gate to the amplitude explicitly.  We see that each $\htt F$ gate introduces a path-independent complex factor of $\frac{1-i}{2 \sqrt{\pi}}$ to the normalization and each $\htt F^{\dag}$ gate a factor of $\frac{1+i}{2 \sqrt{\pi}}$.   Let $\tt q({\rm gate})$ denote the configuration at the input of a gate for path $\gamma$, while $\tt Q({\rm gate})$ denotes the configuration at its output for path $\gamma$. 
In terms of this notation, each $\htt F$ gate introduces a phase factor of $e^{-i \tt Q({\rm gate}) \tt q({\rm gate})}$, each $\htt F^{\dag}$ gate a phase factor of $e^{i \tt Q({\rm gate}) \tt q({\rm gate})}$, and each $\htt P(\eta)$ gate a phase factor of $e^{-i\frac{\eta}{2} \tt q(\rm gate)^2}$.   Finally, we get nontrivial constraints on the allowed paths from  the $\mathds{1}$, $\htt P(\eta)$, $\htt X(\tau)$, $\SUM$ and $\SUM^{\dag}$ gates.  

These results are summarized in the following theorem.  Here, $\sum_{\htt F  \;{\rm gates}}$ denote a sum over all $\htt F$ gates in $\mf C$, and similarly for any other sort of gate, and $\#(\htt F)$ ($\#(\htt F^{\dag})$) denotes the number of $\htt F$ gates ($\htt F^{\dag}$ gates).

\begin{theorem}
\label{CVClPathSumFinal}
For an $n$-system CV Clifford circuit $\mf C$ implementing the overall unitary $\htt U \in \tt C_n$, the transition amplitudes can be computed by the sum-over-paths expression
 \begin{equation}\label{amplitudesumpathsCV}
\langle \bt q_{N} | \htt U | \bt q_0 \rangle = 
\mathfrak{N}_{\mf C}  \int_{ \P_{\mf C}(\bt q_{\rm 0}, \bt q_N)} e^{i \tt S_{\mf C}(\gamma)} \ \dee \gamma,
\end{equation}
where 
$$
\mathfrak{N}_{\mf C} = \left(\frac{1-i}{2 \sqrt{\pi}}\right)^{\#(\hat{F})} \left(\frac{1+i}{2 \sqrt{\pi}}\right)^{\#(\hat{F}^\dag)},
$$
and
\begin{equation}
 \label{CVClActionFinal}
 \tt S_{\mf C}(\gamma) = 
 - \sum_{\htt F  \;{\rm gates}}  \tt q({\rm gate})\tt Q({\rm gate}) 
 + \sum_{\htt F^{\dag} \;{\rm gates}}  \tt q({\rm gate})\tt Q({\rm gate}) 
 - \sum_{\htt P(\eta)\; {\rm gates}} \frac{\eta}{2} \tt q(\rm gate)^2.
  \end{equation}
and where $\int_{\P_{\mf C}(\vec{q}_0, \vec{q}_N)} (\cdot) \dee \gamma$ denotes integration over $\vec{q}_1, \dots, \vec{q}_{N-1}$ subject to the following constraints
\begin{align}\label{allowedpathCV}
&\forall \;\; \mathds{1} , \htt P(\eta) \;{\rm gates} \;: \;\tt Q({\rm gate}) = \tt q({\rm gate}). \;\; \nonumber\\
&\forall \;\; \htt X(\tau) \;{\rm gates} \;: \;\tt Q({\rm gate}) = \tt q({\rm gate}) +\tau. \;\; \nonumber\\
&\forall \;\; \SUM \;{\rm gates} \;:\; \tt Q^{(1)}({\rm gate}) = \tt q^{(1)}({\rm gate}),\;\;\; Q^{(2)}({\rm gate}) = \tt q^{(1)}({\rm gate}) +\tt q^{(2)}({\rm gate}).\;\;\nonumber\\
&\forall \;\; \SUM^{\dag} \;{\rm gates} \;:\; \tt Q^{(1)}({\rm gate}) = \tt q^{(1)}({\rm gate}),\;\;\; Q^{(2)}({\rm gate}) = \tt q^{(2)}({\rm gate}) -\tt q^{(1)}({\rm gate}).
\end{align}  
 \end{theorem}

 What is important for us is the functional form of  $\tt S_{\mf C}(\gamma)$ because we seek to show that it can be interpreted as a generalised action functional through the theory of generating functions.  Before doing so, however, we pause to present a scheme for implementing the constraint on the allowed paths, inspired by the one presented in Dawson {\em et al.}~\cite{Dawson}, and which proceeds by providing an explicit parameterization of the space of allowed paths.   

\red
\black

For the gates $\mathds{1}$, $\htt P(\eta)$, $\htt X(\tau)$, $\SUM$ and $\SUM^\dag$, one sees that for each configuration of the input(s), there is a unique configuration of the output(s) having non-zero amplitude. For the gates $\htt F$ and $\htt F^\dag$, on the other hand, for any configuration of the input, {\em every} configuration of the output has a non-zero amplitude.  For the latter gates, therefore, we must introduce a free parameter for the configuration at their output.


It follows that the number $L$ of parameters sufficient to describe the space of allowed paths is just the sum of the number of $\htt F$ gates and the number of $\htt F^\dag$ gates, $L\equiv \#(\htt F)+\#(\htt F^{\dag}).$
 We will call these the {\em free configuration parameters} and denote them by $\tt x_1, \dots,\tt x_L$, with the collection represented by the vector $\bt x \equiv (\tt  x_1, \dots, \tt x_L)$.  (Note that, unlike elsewhere in this article, the subscript in this notation does not indicate the time-step to which the configuration pertains; it is merely an index for the free parameters.)  It follows that every allowed path can be expressed as a function of these parameters, $\gamma(\bt x)$.  However, the path $\gamma(\bt x)$ may not be an allowed path for all choices of values for the $L$ parameters, and so the parameters are constrained to live in some subspace of $\R^L$.

To visualize the free configuration parameters and the constraints they satisfy, it is useful to annotate the circuit.  The general prescription, which we illustrate with a concrete example in Fig.~\ref{unlabelled}, is  as follows:
\begin{enumerate}
\item Label the configurations of the input systems by $\{\tt q_{0}^{(i)}\}_{i=1}^n$, and the configurations of the output systems by $\{\tt q_{\rm f}^{(i)}\}_{i=1}^n$.
\item For every system immediately following an $\htt F$ gate or an $\htt F^\dag$ gate, 
label the configuration of that system by $\tt x_l$, where $l $ ranges from 1 to $L$, the number of such gates in the circuit.
\item For every system immediately following a $\htt P(\eta)$ or $\htt X(\tau)$ gate, and every pair of systems immediately following a $\SUM$ or $\SUM^\dag$ gate, do not introduce a new label for the configurations of those systems, but rather specify, for each output of the gate, its functional dependence on the inputs of the gate (according to the functional relations determined in Lemma \ref{CVMatElems}).
\end{enumerate}


\begin{figure}[ht]
\resizebox{8cm}{!}{ 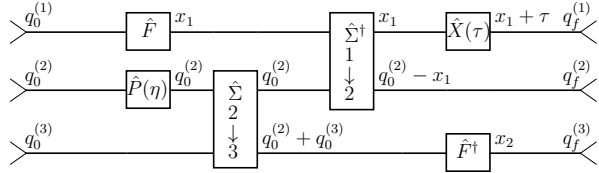 }
\caption{An example of a $3$-system CV Clifford circuit consisting of the following sequence of elementary gates. First, an $\htt F$ gate is implemented on the first system and a $\htt P(\eta)$ gate is implemented on the second system. Next, one has a $\SUM$ gate controlled on the second system and acting on the third system and then a $\SUM^\dag$ gate controlled on the first system and acting on the second system. Finally, one has an $\htt X(\tau)$ gate acting on the first system and an $\htt F^\dag$ gate acting on the third system.  Also indicated is the labelling of the configurations of the systems described in the text. \label{unlabelled}}
\end{figure}

Finally, constraints on the free configuarion parameters arise from the final boundary condition at the output of the circuit.
For every $i \in \{1,\dots, n\}$, define $\tt B^{(i)}(\bt x)$
to be the configuration of the $i$th system at the output of the circuit as a function of the free configuration parameters, $\bt x$.  The form of this function can depend on the configurations of the input systems, $\bt q_0$, which are given as initial conditions, as well as the $\tau$ and $\eta$ parameters of the $X(\tau)$ and $P(\eta)$ gates which are given by the specification of the circuit. In our example of Fig.~\ref{unlabelled}, for instance,
\begin{eqnarray}
\tt B^{(1)}(\bt x) &=& \tt x_1+\tau,\nonumber\\
\tt B^{(2)}(\bt x) &=& \tt q_0^{(2)} - \tt x_1,\nonumber\\
\tt B^{(3)}(\bt x) &=& \tt x_2.\nonumber\\
\end{eqnarray}

For a general circuit $\mf C$, the vector of free configuration parameters, $\vec x$, is constrained to the set $\mathbb{F}_{\mf C}(\bt q_0,\bt q_f) $, where 
\be{\label{CVClAllowedFinal} 
\mathbb{F}_{\mf C}(\bt q_0,\bt q_f) \equiv  \left\{ \bt x \ | \ \tt B^{(i)}(\bt x) = \tt q_{\rm f}^{(i)}, \ \forall i \in \{1,\dots, n\}\right\}.}

In general, each constraint equation on $\bt x$ defines an affine hyperplane in $\R^{L}$. As such, $\mathbb{F}_{\mf C}(\bt q_0,\bt q_f)$ describes the (possibly empty) intersection of these affine hyperplanes. 

For the example of Fig.~\ref{unlabelled}, for instance, the set is
\begin{equation}\label{allowedpathsexample}
 \{ \bt x \ | \ \tt x_1 + \tau = \tt q_{\rm f}^{(1)}, \ \tt q_0^{(2)} - \tt x_1 = \tt q_{\rm f}^{(2)}, \ \tt x_2  = \tt q_{\rm f}^{(3)} \}. 
 \end{equation}\black
  Note that in our example, the free configuration parameter $\tt x_2$ is fixed directly by the final boundary condition, so that one need not have introduced it.  Indeed, one can restrict the free configuration parameters to the systems that are at the output of {\em nonterminal} $\htt F$ and $\htt F^\dag$ gates (where nonterminal means not the last gate acting on a given system).  This does not, however, change the complexity of solving the constraints.

 Given this parameterization of the allowed paths, we can rewrite \eq{amplitudesumpathsCV} as
 \begin{equation}
\langle \bt q_{\rm f} | \htt U | \bt q_0 \rangle = 
\mathfrak{N}_{\mf C}  \int_{\mathbb{F}_{\mf C}(\bt q_0,\bt q_f)} e^{i \tt S_{\mf C}(\gamma(\bt x))} \ \dee^L x,
\end{equation}
where $\int_{\mathbb{F}_{\mf C}(\bt q_0,\bt q_f)} (\cdot) \dee^L x$ denotes integration over the subspace given as the intersection of affine hyperplanes within $\R^L$ that are picked out by the constraints on $x_1, \dots, x_L$ in the definition of $\mathbb{F}_{\mf C}(\bt q_0,\bt q_f)$.\black



Within this integral, the phase of an allowed path is specified as a function of the free parameters $\bt x$ by adapting the functional form of \eq{CVClActionFinal} to the labelling scheme described above.  For instance, in our example, the phase of the path detemined by free parameters $\bt x$ is 
$$S_{\mf C}(\gamma(\bt x))= -\tt q_0^{(1)} \tt x_1 + (\tt q_0^{(2)}+\tt q_0^{(3)})\tt x_2 - \frac{\eta}{2} (\tt q_0^{(2)})^2.$$\black


\subsection{The discrete-time analogue of the action functional for CV systems}
\label{genfunctions}
Let $\Omega^j(\R^{2n})$ denote the vector space \black of all $j$-forms on the phase space $\R^{2n}$, and let $\Omega(\R^{2n}) = \oplus_{j=0}^{2n} \Omega^j(\R^{2n})$ denote the algebra of {\em all} differential forms on $\R^{2n}$. 
Introducing canonical coordinates $({\bt q},{\bt p})$, the $2$-form $\omega \in \Omega^2(\R^{2n})$ defined by
\begin{equation}
\omega= \sum_{i=1}^n\dee \tt q^{(i)} \land \dee \tt p^{(i)}
\end{equation}
is a symplectic form because it is non-degenerate and $\dee \omega = 0$. A smooth function $\phi:\R^{2n} \to \R^{2n}$, $({\bt q},{\bt p}) \mapsto ({\bt Q}, {\bt P})$ is said to be a symplectomorphism if
\begin{equation}
\label{def:symp}
\sum_{i=1}^n\dee \tt q^{(i)} \land \dee \tt p^{(i)} = \sum_{i=1}^n \dee \tt Q^{(i)} \land \dee \tt P^{(i)}.
\end{equation}

There is a canonical $1$-form, 
\begin{equation}
\theta = \sum_{i=1}^n \tt p^{(i)} \dee \tt q^{(i)}.
\end{equation}
which satisfies $\omega = - \dee \theta$.
One can restate the condition for $\phi$ to be a symplectomorphism, \eq{def:symp}, in terms of this canonical $1$-form $\theta$: $\phi$ is a symplectomorphism if and only if there is a $\tilde{\tt G}({\bt q}, {\bt p}) \in C^\infty(\R^{2n})$ such that
\begin{equation}
\label{taut}
\sum_{i=1}^n \tt P^{(i)} \dee \tt Q^{(i)} - \sum_{i=1}^n \tt p^{(i)} \dee \tt q^{(i)} = \dee \tilde{\tt G}({\bt q},{\bt p}).
\end{equation}
We call such a $\tilde{\tt G}(\bt q,\bt p)$ a {\em generating function} associated to the symplectomorphism $\phi$ (\cite{Arnold}, Chapter 9). We note that the existence of a generating function for every symplectomorphism depends on the fact that every closed $1$-form on $\R^{2n}$ is exact. Note that generating functions are only unique up to addition of scalars.

If the ${\bt q}$ and ${\bt Q}$ variables can be taken to be independent, then one can express $\tilde{\tt G}({\bt q},{\bt p})$ purely in terms of ${\bt q}$ and ${\bt Q}$, i.e., ${\tt G}({\bt q},{\bt Q}):=\tilde{\tt G}({\bt q},{\bt p}({\bt q},{\bt Q}))$. It then follows from \eq{taut} that
\begin{equation}
\tt p^{(i)} = - \frac{\partial \tt G}{\partial \tt q^{(i)}}, \quad \tt P^{(i)} = \frac{\partial \tt G}{\partial \tt Q^{(i)}},
\label{genfndefn}
\end{equation}
which is the sense in which $\tt G({\bt q},{\bt Q})$ generates the symplectomorphism $\phi$. Because ${\bt q}$ and ${\bt Q}$ can indeed be taken to be independent  in all of the cases we will consider, whenever we refer to the generating function, we mean ${\tt G}({\bt q},{\bt Q})$.

Note that for a symplectomorphism $\phi:(\bt q,\bt p) \mapsto (\bt Q,\bt P)$ that results from a continuous-time Hamiltonian dynamics acting over a finite time interval, the generating function of that symplectomorphism, $G(q,Q)$, is simply the action of the classical trajectory over that time interval which has $q$ as the initial configuration and $Q$ as the final configuration, as in \eq{generatingfunction}.   


We now apply the proposal of Eq.~\eqref{CVactionfunctional} from the introduction, namely, that the analogue of the action functional for discrete-time dynamics is the sum of the generating functions associated to the symplectomorphisms that make up the discrete-time dynamics. 

Consider an $n$-system CV Clifford circuit with $N$ time-steps, so that the space of paths is $\R^{n(N+1)}$.  Let $\Phi = \{ \phi_k \}_{k=1}^N$ be the sequence of symplectomorphisms of $\rx^{2n}$ associated to the circuit, and denote the generating function associated to $\phi_k$ by\footnote{The fact that the $G_{\phi_k}(\vec{q}_{k-1},\vec{q}_{k})$ can be written as real polynomials over $\vec q_{k-1}$ and $\vec q_k$ will be demonstrated explicitly in Lemma~\ref{CVClElemGen}.}
\begin{equation}
G_{\phi_k}(\vec{q}_{k-1},\vec{q}_{k}) \in \R[\vec q_{k-1}, \vec q_k].
\end{equation}
The definition of an action functional for discrete-time paths of CV systems, proposed in  \eq{CVactionfunctional} of the introduction, specifies that the action functional for this circuit should be as follows.
\begin{definition}
\label{def:genfunctionalcont}
The action functional over paths in $\R^{n(N+1)}$ that is associated to the sequence of symplectomorphisms $\Phi = \{\phi_k\}_{k=1}^N$, denoted $\gamma \mapsto S_{\Phi}(\gamma)$,  is 
\begin{equation}\label{proposalSPhi}
\tt S_{\Phi} (\bt q_0,\ldots, \bt q_N) \equiv \sum_{k=1}^N \tt G_{\phi_k}(\bt q_{k-1},\bt q_k).
\end{equation}
\end{definition}

Note that $\tt S_{\Phi}$ defines a functional on the set of discrete-time paths through configuration space. Given any such path $\gamma = (\bt q_0,\ldots, \bt q_N)$, $\tt S_\Phi$ can be evaluated on $\gamma$ to obtain the real number $\tt S_\Phi(\gamma)$ which is the {\em total action} of the path $\gamma$ under the dynamics $\Phi$.

As discussed in the introduction, at Eq. \ref{generatingfunction}, for continuous-time dynamics,
the evaluation of the action functional on the classical trajectory yields the generating function of the overall symplectomorphism encoding the dynamics. As will now be demonstrated, the same property holds true for the discrete-time action functional of Definition \ref{def:genfunctionalcont}.

Given a path $\Gamma = (\bt q_0, \bt p_0 \ldots, \bt q_N, \bt p_N) \in (\R^{2n})^{N+1}$ through phase space one can extract a path $\Gamma^q = (\bt q_0,\ldots, \bt q_N) \in (\R^n)^{N+1}$ through configuration space.  For each choice of initial phase space value $\Gamma_0=(\bt q_0, \bt p_0)$, a sequence of symplectomorphisms $\Phi = \{\phi_k\}_{k=1}^N$ defines a canonical path through phase space 
\begin{equation}
 \Gamma^{\rm cl}_\Phi (\bt q_0, \bt p_0) = \left(\Gamma_0, \Phi_1(\Gamma_0), \Phi_2 (\Gamma_0), \ldots, \Phi_N (\Gamma_0) \right),
\end{equation}
where $\Phi_i = \phi_i \circ \cdots \circ \phi_1$ is the overall symplectomorphism implemented after the $i$-th time-step. We then define a {\em classical trajectory having initial configuration $\bt q_0$ and final configuration $\bt q_N$} to be a path through configuration space 
\begin{equation}
 \gamma^{\rm cl}_\Phi(\bt q_0, \bt q_N) = (\Gamma^{\rm cl}_\Phi)^q(\bt q_0, \bt p_0)
\end{equation}
for some choice of $\bt p_0$ such that the $q$-component of $\phi_N \phi_{N-1} \cdots \phi_1(\Gamma_0)$ equals $q_N$.

Note that for specific dynamics $\Phi$ and boundary values $\bt q_0$ and $\bt q_N$, $\gamma_{\rm cl}^\Phi(\bt q_0, \bt q_N)$ may not exist, and when it does exist it is not necessarily unique.\footnote{Examples are found in the next section: 
for the single symplectomorphism $\phi_{\hat{F}}$, for a given choice of initial and final boundary conditions on the configuration space, there are many choices of the initial momentum that are consistent with the symplectomorphism, and for  the single symplectomorphism $\phi_{\hat{\Sigma}}$, it is straightforward to find examples of initial and final boundary conditions on the configuration space which are inconsistent with this symplectomorphism.}


\begin{proposition}
\label{actclasstraj}
For a sequence of symplectomorphisms $\Phi = \{\phi_k\}_{k=1}^N$, let a classical discrete-time trajectory through a continuous-variable configuration space $\R^{n}$ 
having initial configuration $\bt q_0$ and final configuration $\bt q_N$ be denoted $\gamma^{\rm cl}_\Phi(\bt q_0, \bt q_N)$.  The action functional of Definition~\ref{def:genfunctionalcont}, evaluated on $\gamma^{\rm cl}_\Phi(\bt q_0, \bt q_N)$, yields  exactly the generating function of the overall symplectomorphism $\Phi_N$;  $S_\Phi(\gamma^{\rm cl}_\Phi(\bt q_0, \bt q_N)) = G_{\Phi_N}(\bt q_0, \bt q_N)$.
\end{proposition}
\begin{proof}

For simplicity of presentation, we will work with the action functional defined on discrete-time paths through {\em phase space} rather than through configuration space. This results in no loss of generality as this is in fact the most natural domain of definition (eq. \ref{taut}). The action functional for paths through phase space is given by
\begin{equation}
 \tt \tilde{S}_\Phi (\bt q_0, \bt p_0, \ldots, \bt q_{N-1}, \bt p_{N-1}) \equiv \sum_{k=1}^N \tt \tilde{G}_{\phi_k}(\bt q_{k-1}, \bt p_{k-1}).
\end{equation}
In a similar manner to how $\tt S_\Phi$ can be evaluated for paths through configuration space, $\tt \tilde{S}_\Phi$ can be evaluated on paths $\Gamma = (\bt q_0, \bt p_0, \ldots, \bt q_N, \bt p_N)$ through phase space to obtain a real number $\tt \tilde{S}_\Phi(\Gamma)$. Let $q$ and $p$ superscripts denote, respectively, the position and momentum components of a phase-space object, so that  $\Gamma^q$ is the path through configuration space defined by the path $\Gamma$ through phase space.  Note that $\tilde{S}_\Phi(\Gamma) = S_\Phi(\Gamma^q)$ since for any symplectomorphism $\phi:(\bt q, \bt p) \mapsto (\bt Q, \bt P)$ for which $q$ and $Q$ can be taken to be independent, $\tilde{G}_\phi(\bt q, \bt p) = G_\phi(\bt q, \bt Q(\bt q, \bt p))$.

 Making use of the shorthand $\bt x \dee \bt y \equiv \sum_{i=1}^n x^{(i)} \dee y^{(i)}$, one has the following
\begin{eqnarray*}
& & \dee \tilde{G}_{\Phi_{N}} (\bt q_0, \bt p_0) \equiv \bt \Phi_{N}(\Gamma_0)^p \dee \bt \Phi_{N}(\Gamma_0)^q - \bt p_0 \dee \bt q_0 \\
 &=& \left( \bt \Phi_{N}(\Gamma_0)^p \dee \bt \Phi_{N}(\Gamma_0)^q - \bt \Phi_{N-1}(\Gamma_0)^p \dee \bt \Phi_{N-1}(\Gamma_0)^q\right) + \cdots + \left( \bt \Phi_{1}(\Gamma_0)^p \dee \bt \Phi_{1}(\Gamma_0)^q - \bt p_0 \dee \bt q_0\right) \\
 &=& \dee \tilde{G}_{\phi_{N}}(\Phi_{N-1}(\Gamma_0)) + \dee \tilde{G}_{\phi_{N-1}}(\Phi_{N-2}(\Gamma_0)) + \cdots + \dee \tilde{G}_{\phi_1}(\Gamma_0) \\
 &\equiv& \dee \tilde{S}_\Phi\left(\Gamma^{\rm cl}_\Phi (\bt q_0, \bt p_0)\right).
\end{eqnarray*}
By assumption one has that $\tilde{G}_{\Phi_N}(\bt q_0, \bt p_0) = G_{\Phi_N}(\bt q_0, \bt q_N)$ and $\tilde{S}_\Phi\left(\Gamma^{\rm cl}_\Phi (\bt q_0, \bt p_0)\right) = S_\Phi(\gamma^{\rm cl}_\Phi(\bt q_0, \bt q_N))$, and so the above calculation proves that $G_{\Phi_N}(\bt q_0, \bt q_N)$ and $S_\Phi(\gamma^{\rm cl}_\Phi(\bt q_0, \bt q_N))$ differ by the addition of a constant. Since the generating function is itself only defined up to this same ambiguity the claimed result follows.
\end{proof}

We now turn to determining the explicit form of the generating functions associated to  each of the gates in the generating set of the CV Clifford group.


\subsection{Symplectomorphisms and generating functions for CV Clifford gates}
\label{CVGenAction}
To see that our proposal for the action functional of a path, \eq{proposalSPhi}, does indeed yield the functional that appears in the path-sum expression, \eq{CVClActionFinal}, we must specify the classical counterpart of each gate in the generating set of the CV Clifford group, that is, both the classical symplectomorphism associated to it and the generating function of this symplectomorphism. 
For an $n$-system CV Clifford circuit, it is natural to take $\R^{2n}$ with basis $(\tt q^{(1)}, \ldots, \tt q^{(n)}, \tt p^{(1)}, \ldots, \tt p^{(n)})$ and the usual symplectic form as the underlying symplectic manifold.  Looking at the expressions for the unitaries associated to the elementary gate set for CV Clifford circuits, \eq{CVGateSet}, one sees that a Hamiltonian operator can be associated to each gate. 
Each such Hamiltonian operator defines a classical Hamiltonian function on $\rx^2$ (or $\rx^4$ for the $\SUM$ gate) in the usual way, and thus a symplectomorphism on $\rx^2$ (or $\rx^4$).  (Note that one can equally well deduce the symplectomorphism associated to each CV Clifford gate by determining its Wigner representation.)  

\begin{lemma}
\label{CVClElemSymp}
  The elementary CV Clifford gates of \eq{CVGateSet} are associated to the following symplectomorphisms:
 \begin{itemize}
  \item $\phi_{\htt F}: (\tt q, \tt p) \mapsto(\tt p, - \tt q)$;
  \item $\phi_{\htt F^\dag}:( \tt q, \tt p) \mapsto (-\tt p, \tt q)$;
  \item $\phi_{\htt P(\eta)}: ( \tt q, \tt p) \mapsto (\tt q, \tt p - \eta \tt q)$;
  \item $\phi_{\htt X(\tau)}: (\tt q, \tt p) \mapsto (\tt q +\tau, \tt p)$;
  \item $\phi_{\SUM}:(\tt q^{(1)}, \tt q^{(2)}, \tt p^{(1)}, \tt p^{(2)}) \mapsto (\tt q^{(1)}, \tt q^{(2)} + \tt q^{(1)}, \tt p^{(1)} - \tt p^{(2)}, \tt p^{(2)})$;
  \item $\phi_{\SUM^\dag}:(\tt q^{(1)}, \tt q^{(2)}, \tt p^{(1)}, \tt p^{(2)}) \mapsto (\tt q^{(1)}, \tt q^{(2)} - \tt q^{(1)}, \tt p^{(1)} + \tt p^{(2)}, \tt p^{(2)})$.
 \end{itemize}
\end{lemma}
These will be called the elementary CV Clifford symplectomorphisms.
\begin{proof}
\begin{itemize}
\item The $\htt F$ gate corresponds to evolution under the Hamiltonian $H_{\htt F} (\tt q, \tt p) =  \frac{\pi}{4}(\tt p^2 + \tt q^2)$ for a time interval of unit duration.  
Solving Hamilton's equation for this Hamiltonian yields,
\begin{equation}
 \tt q(t) = \tt q(0) \cos\left( \frac{\pi t}{2} \right) + \tt p(0) \sin\left( \frac{\pi t}{2} \right), \quad \tt p(t) = \tt p(0) \cos\left( \frac{\pi t}{2} \right) - \tt q(0)\sin\left( \frac{\pi t}{2} \right),
\end{equation}
and substituting $t=1$ yields
\be{\phi_{\htt F}: (\tt q, \tt p) \mapsto(\tt p, - \tt q).}
Similarly, the $\htt F^{\dag}$ gate corresponds to evolution under the negative of this Hamiltonian for a unit time interval, and so
\be{\phi_{\htt F^\dag}:( \tt q, \tt p) \mapsto (-\tt p, \tt q).}
\item The gate $\htt P(\eta)$ corresponds to evolution under the Hamiltonian $H_{\htt P}(\tt q, \tt p) =  \frac{1}{2} \tt q^2$ for a time interval of duration $\eta$.
Solving Hamilton's equations for this Hamiltonian yields,
\begin{equation}
\tt q(t) = \tt q(0), \quad \tt p(t) =  -\tt q(0) t + \tt p(0),
\end{equation}
and substituting $t=\eta$ gives
\be{\phi_{\htt P(\eta)}: ( \tt q, \tt p) \mapsto (\tt q, \tt p - \eta \tt q).}
\item The gate $\htt X(\tau)$ corresponds to evolution under the Hamiltonian $H_{\htt X}(\tt q, \tt p) = \tt p$ for a time interval of duration $\tau$.
Solving Hamilton's equations for this equation yields,
\begin{equation} 
\tt q(t) = t + \tt q(0), \quad \tt p(t) = \tt p(0),
 \end{equation}
and substituting $t=\tau$ gives
\be{\phi_{\htt X(\tau)}: (\tt q, \tt p) \mapsto (\tt q +\tau, \tt p).}
\item Finally, the $\SUM$ gate corresponds to evolution under the Hamiltonian $H_{\SUM} (\bt q, \bt p) =  \tt q^{(1)} \tt p^{(2)}$ for a time interval of unit duration.
Again, solving Hamilton's equations yields,
\begin{equation}
\begin{array}{cc}
\tt q^{(1)}(t) = \tt q^{(1)}(0) & \tt q^{(2)}(t) = \tt q^{(1)}(0) t + \tt q^{(2)} (0) \\
\tt p^{(1)}(t) =  -\tt p^{(2)}(0) t + \tt p^{(1)}(0) & \tt p^{(2)}(t) = \tt p^{(2)}(0)
\end{array},
\end{equation}
and substituting $t=1$ gives
\be{\phi_{\SUM}:(\tt q^{(1)}, \tt q^{(2)}, \tt p^{(1)}, \tt p^{(2)}) \mapsto (\tt q^{(1)}, \tt q^{(2)} + \tt q^{(1)}, \tt p^{(1)} - \tt p^{(2)}, \tt p^{(2)}).}
The $\SUM^{\dag}$ gate corresponds to evolution under the negative of this Hamiltonian for the same duration, so that
\be{\phi_{\SUM^\dag}:(\tt q^{(1)}, \tt q^{(2)}, \tt p^{(1)}, \tt p^{(2)}) \mapsto (\tt q^{(1)}, \tt q^{(2)} - \tt q^{(1)}, \tt p^{(1)} + \tt p^{(2)}, \tt p^{(2)}).}
\end{itemize}
\end{proof}

We now turn to the generating functions associated to each of these symplectomorphisms.

\begin{lemma}
\label{CVClElemGen}
The following are the generating functions of the elementary CV Clifford symplectomorphisms:
\begin{itemize}
\item $\tt G_{\phi_{\htt F}} (\tt q, \tt Q) = - \tt q \tt Q;$
\item $\tt G_{\phi_{\htt F^\dag}} (\tt q, \tt Q) = \tt q \tt Q;$
\item $\tt G_{\phi_{\htt P(\eta)}} (\tt q, \tt Q) = -\frac{\eta}{2} \tt q^2;$
\item $\tt G_{\phi_{\htt X(\tau)}} (\tt q, \tt Q) = 0;$
\item $\tt G_{\phi_{\SUM}}(\tt q^{(1)}, \tt q^{(2)}, \tt Q^{(1)}, \tt Q^{(2)}) = 0;$
\item $\tt G_{\phi_{\SUM^\dag}}(\tt q^{(1)}, \tt q^{(2)}, \tt Q^{(1)}, \tt Q^{(2)}) = 0.$
\end{itemize}
\end{lemma}
\begin{proof}
These follow by direct computation using the definition in \eq{taut}
\begin{itemize}
\item For the $\htt F$ gate,
\begin{equation}
\tt P \dee \tt Q - \tt p \dee \tt q = -\tt q \dee \tt p - \tt p \dee \tt q = \dee ( - \tt q \tt p) = \dee ( -\tt q \tt Q).
\end{equation}
\item For the $\htt F^\dag$ gate,
\begin{equation}
\tt P \dee \tt Q - \tt p \dee \tt q = \tt q \dee (- \tt p) - \tt p \dee \tt q = \dee (- \tt q \tt p) = \dee ( \tt q \tt Q).
\end{equation}
\item For the $\htt P(\eta)$ gate,
\begin{equation}
\tt P \dee \tt Q - \tt p \dee \tt q = (\tt p - \eta \tt q) \dee \tt q - \tt p \dee \tt q = - \eta \tt q \dee \tt q = \dee \left(- \frac{\eta}{2} \tt q^2\right).
\end{equation}
\item For the $\htt X(\tau)$ gate,
\begin{equation}
\tt P \dee \tt Q - \tt p \dee \tt q = \tt p \dee ( \tt q + \tau) - \tt p \dee \tt q = \dee (0).
\end{equation}
\item For the $\SUM$ gate,
\begin{align}
 &     \tt P^{(1)} \dee \tt Q^{(1)} + \tt P^{(2)} \dee \tt Q^{(2)}- \tt p^{(1)} \dee \tt q^{(1)} - \tt p^{(2)} \dee \tt q^{(2)} \nonumber \\
&= (\tt p^{(1)} - \tt p^{(2)}) \dee \tt q^{(1)} + \tt p^{(2)} \dee (\tt q^{(2)} + \tt q^{(1)}) - \tt p^{(1)} \dee \tt q^{(1)} - \tt p^{(2)} \dee \tt q^{(2)} \nonumber \\
&= \dee (0).
\end{align}
\item For the $\SUM^\dag$ gate,
\begin{align}
 &     \tt P^{(1)} \dee \tt Q^{(1)} + \tt P^{(2)} \dee \tt Q^{(2)}- \tt p^{(1)} \dee \tt q^{(1)} - \tt p^{(2)} \dee \tt q^{(2)} \nonumber \\
&= (\tt p^{(1)} + \tt p^{(2)}) \dee \tt q^{(1)} + \tt p^{(2)} \dee (\tt q^{(2)} - \tt q^{(1)}) - \tt p^{(1)} \dee \tt q^{(1)} - \tt p^{(2)} \dee \tt q^{(2)} \nonumber \\
&= \dee (0).
\end{align}
\end{itemize}
\end{proof}

 Note that instances of the identity gate correspond to the identity symplectomorphism $( \tt q, \tt p) \mapsto (\tt q, \tt p)$ and have generating function equal to 0.  Note also that if gates act in parallel on different systems, the symplectomorphism for the overall gate is simply the composition of the symplectomorphisms of the component gates. and the generating function for the overall gate is simply the sum of the generating functions of the component gates.

\subsection{Main result}

Let $\mf C$ be an $n$-system CV Clifford circuit consisting of $N$ time-steps and wherein all the gates are CV Clifford gates.  To such a circuit, there is an associated sequence  of symplectomorphisms of $\rx^{2n}$, denoted $\Phi=\{ \phi_k\}_{k=1}^N$, where each of the $\phi_k$ is composed of the elementary CV Clifford symplectomorphisms described in Lemma \ref{CVClElemSymp}. \black Then, according to our Definition \ref{def:genfunctionalcont}, the action functional over discrete-time paths 
associated to $\Phi$, denoted $\gamma \mapsto \tt S_{\Phi}(\gamma)$,
 is the sum of the generating functions associated to these symplectomorphisms. 
Specifically, Lemma  \ref{CVClElemGen} implies that
\begin{equation}
 \tt S_{\Phi}(\gamma) = 
 - \sum_{\htt F  \;{\rm gates}}  \tt q({\rm gate})\tt Q({\rm gate}) 
 + \sum_{\htt F^{\dag} \;{\rm gates}}  \tt q({\rm gate})\tt Q({\rm gate}) 
 - \sum_{\htt P(\eta)\; {\rm gates}} \frac{\eta}{2} \tt q(\rm gate)^2,
  \end{equation}
 where we have adopted the notational convention introduced above Theorem \ref{CVClPathSumFinal}.
Comparison with  \eq{CVClActionFinal} then establishes our main result for CV Clifford circuits.

\begin{theorem}
\label{MainCV}
Consider an $n$-system CV Clifford circuit $\mf C$, associated in quantum theory with a unitary $\htt U \in C_n$ 
and associated, in its classical counterpart, to a symplectomorphism $\Phi$.
The functional $\tt S_{\Phi}(\gamma)$ that specifies, via Definition~\ref{def:genfunctionalcont}, the action of the discrete-time path $\gamma$
through the classical counterpart of the circuit 
 is precisely equal to the functional $\tt S_{\mf C}(\gamma)$ that defines the phase assigned to $\gamma$ in the sum-over-paths expression for the transition amplitude of the quantum circuit, \eq{PathSumExpression}.
\end{theorem}

\black

\section{Quopit Clifford circuits}
\label{Clifford}

\subsection{Sum-over-paths expression for quopit Clifford circuits}
\label{PathSum}
 We turn now to quopit Clifford circuits.  Clifford circuits for collections of discrete systems of arbitrary dimension $d$ were first introduced by Gottesman in \cite{GottesmanCliff} as a higher dimensional version of the qubit stabiliser codes for fault-tolerant quantum computation (and it was shown that the Gottesman-Knill theorem extends to these higher dimensions).  A qudit of dimension $d$ equal to an odd prime has been termed a ``quopit''~\cite{Emerson}.  Hence, quopit Clifford circuits are Clifford circuits wherein the elementary systems are dimension $d$ for $d$ an odd prime.

This Section will follow a structure similar to Section \ref{CVClifford}. We begin by determining a sum-over-paths expression for transitions amplitudes of quopit Clifford circuits, as in \eq{PathSumExpression},  thus identifying the functional over paths appearing in the exponent of the phase factor.  We then address the question of whether this functional
admits of an interpretation in terms of a generalised action functional, just as was done in Sections \ref{genfunctions} and \ref{CVGenAction}. A discrete phase space representation of quopit Clifford circuits is described in Section \ref{Symprep}, 
In Section \ref{genfundisc}, it is shown how to define generating functions for symplectomorphisms on a discrete phase space using tools from algebraic geometry. 
 In Section \ref{DiscGenAction},  the symplectomorphisms associated to the gates in the elementary gate set are identified, and, using the tools of Section \ref{genfundisc}, we find the associated generating functions. 
 Finally, we show that an action functional defined via the sum of these generating functions coincides with the functional appearing in the exponent of the phase factor for the sum-over-paths expression of the circuit dynamics.

An $n$-quopit Clifford circuit consists of preparations and measurements in the computational basis of $(\C^d)^{\otimes n}$, where $d$ is an odd prime, and has elementary gate set
\be{\label{def:CliffGateSet}
\hat{F} &=& \frac{1}{\sqrt{d}} \sum_{q,q' \in \Z_d} \chi(qq') |q \rangle \langle q'| \nn
\hat{R} &=& \sum_{q \in \Z_d} \chi(2^{-1}q(q-1)) |q \rangle \langle q| \nn
\Sig &=& \sum_{q,q' \in \Z_d} |q,q+q' \rangle \langle q,q'|,
}
where 
\be{ \chi(p) \equiv e^{\frac{2 \pi i p}{d}}}
 and arithmetic operations on elements of $\Z_d$ are done modulo $d$. We call $\hat{F}$ the Fourier gate, $\hat{R}$ the Phase gate and $\Sig$ the Sum gate.

It has been shown by Clark  \cite{Clark} that this gate set can (up to a global phase factor) implement any unitary lying in the $n$-quopit Clifford group, which we denote $C_{d,n}$. To define this group, we must introduce the $n$-quopit Pauli group, denoted ${\cal G}_{d,n}$, the $d$-dimensional generalization of the qubit Pauli group. This is the subgroup of $U\left((\C^d)^{\otimes n} \right)$ generated by $\{ \chi(\htt q_i), \chi(\htt p_i): i\in \{1,\dots,n\}\}$ and $e^{\frac{i \pi}{d}}\mathds{1}$, where $i$ labels the quopits and where for a given quopit, 
\begin{equation}
 \chi(\hat q) \equiv \sum_{q\in \Z_d} \chi(q) \kbra qq, \quad \chi(\hat p) \equiv \sum_{q \in \Z_d} \kbra{q+1}q.
\end{equation}
\begin{definition}
\label{def:Cliffgroup}
The $n$-quopit Clifford group, $C_{d,n}$, is defined to be the normaliser of the $n$-quopit Pauli group ${\cal G}_{d,n}$ inside $U\left((\C^d)^{\otimes n} \right)$, that is, $C_{d,n} := N\left({\cal G}_{d,n}\right)$.
\end{definition}

Note that $C_{d,n}$ has $U(1)$ as a subgroup given by the operators $e^{i \phi} {\mathds 1}$. Clark has proven~\cite{Clark} that the set
\begin{equation}
 \left\{ \hat F_i,\hat R_i, \hat \Sigma_{i,j} \ : \ i, j \in \{1,\ldots, n\} \right\}
\end{equation}
are a generating set for the group $ C_{d,n} / U(1)$.

Let ${\mf C}$ be a given $n$-quopit Clifford circuit implementing a unitary $\hat{U} \in C_{d,n}$. To calculate the corresponding transition amplitudes one must first know the matrix elements for each of the elementary gates. 
\begin{lemma}
The matrix elements for the elementary quopit Clifford gates are:
\begin{itemize}
 \item $\langle Q | \hat{F} | q \rangle =\frac{1}{\sqrt{d}} \chi(q Q)$;
 \item $\langle Q | \hat{R} | q \rangle = \chi( 2^{-1} q(q-1)) \delta_{Q,q}$;
 \item $\langle Q^{(1)}, Q^{(2)} | \hat{\Sigma} | q^{(1)}, q^{(2)} \rangle = \delta_{Q^{(1)},q^{(1)}} \delta_{Q^{(2)}, q^{(1)}+q^{(2)}}$.
\end{itemize}
It is evident, therefore, that these gates are balanced.
\end{lemma}

 These identities are straightforward to verify.  

If, at some time-step, a quopit has no gate acting on it, we shall describe the gate as {\em identity} and denote it by $\mathds{1}$.  The identity gate can be obtained by acting the Fourier gate twice in succession, so that one can infer from lemma \ref{CVMatElems} and a short calculation that its contibution to the amplitude is what one expects, namely, $ \delta_{Q,q}$

Following argumentation parallel to that provided in Section \ref{CVClifford}, except where the variables take values in $\Z_d$ as opposed to $\R$, we obtain the following result.
\begin{theorem}
\label{penult}
 Given an $n$-quopit Clifford circuit ${\mf C}$ implementing a unitary $\hat{U} \in C_{d,n}$ the transition amplitudes can be computed by the sum-over-paths expression,
 \begin{equation}
\langle \vec q_N| \hat{U} | \vec{q}_0 \rangle = \mathfrak{N}_{\mf C} \sum_{\gamma \in \P_{{\mf C}}(\vec q_0, \vec q_N)} e^{i S_{{\mf C}}(\gamma)},
\end{equation}
where
\begin{equation}
\mathfrak{N}_{\mf C} = \frac{1}{d^{\#(\hat{F})/2}},
\end{equation}
where
\begin{equation}\label{quopitphasefactor}
S_{{\mf C}}( \gamma) =  \frac{2\pi}{d}  \left( - \sum_{\hat{F} \ {\rm gates}} q({ \rm gate}) Q({\rm gate}) + \sum_{\hat{R} \ {\rm gates}} 2^{-1} q({\rm gate}) \left( q({\rm gate}) - 1\right) \right),
\end{equation}
 and where the set of allowed paths is given by $\P_{{\mf C}}( \vec q_0,\vec q_f)$, defined as the set of paths satisfying the following constraints
\begin{align}\label{Kroneckerdiscrete}
&\forall \;\; \mathds{1} , \htt R \;{\rm gates} \;: \;\tt Q({\rm gate}) = \tt q({\rm gate}) \;\; \nonumber\\
&\forall \;\; \hat{\Sigma} \;{\rm gates} \;:\; \tt Q^{(1)}({\rm gate}) = \tt q^{(1)}({\rm gate}),\;\;\; Q^{(2)}({\rm gate}) = \tt q^{(1)}({\rm gate}) +\tt q^{(2)}({\rm gate}).
\end{align}  
\end{theorem}

For the remainder of this Section, we will show that 
$S_{{\mf C}}(\gamma)$ can be understood as a generalised action functional.

As before,  however, we pause here to describe a method of implementing the constraint to the allowed paths in terms of a parameterization.  We denote the initial configurations by $\vec q_0$ and the final configurations by $\vec q_{\rm f}$ and for each $\hat{F}$ gate in the circuit we introduce a free configuration parameter at its output.
We denote the free configuration parameters as $\vec x \equiv (x_1, \ldots, x_L)$, where $L = \#(\hat{F})$. Letting $B^{(i)}(\vec{x})$ denote the configuration of the $i$-th system at the output of the circuit, 
one has that the space of allowed values of  $\vec x$ is
\begin{equation}
\mathbb{F}_{{\mf C}}( \vec q_0,\vec q_f) =\left \{ \vec x \ | \ B^{(i)}(\vec x) = q_f^{(i)}, \ \forall i \in \{1, \ldots, n\} \right\}.
\end{equation}
Just as was found for the continuous case, each of the above equations defines an affine hypersurface in $\Z_d^L$, and so $\mathbb{F}_{{\mf C}}(\vec q_0, \vec q_f)$ is a (possibly empty) subset of $\Z_d^L$ given by the intersection of $n$ affine hypersurfaces.  Given this parameterization, we have
 \begin{equation}
\langle \bt q_{\rm f} | \htt U | \bt q_0 \rangle = 
\mathfrak{N}_{\mf C}  \sum_{\vec x \in \mathbb{F}_{\mf C}(\bt q_0,\bt q_f)} e^{i \tt S_{\mf C}(\gamma(\bt x))}.
\end{equation}

\subsection{Symplectic representation of discrete systems}
\label{Symprep}
Consider the vector space $(\zx_d)^{2n}$ with basis $(q^{(1)},\ldots, q^{(n)}, p^{(1)}, \ldots, p^{(n)})$. One can introduce a symplectic inner product on this space in the usual fashion: Letting 
\begin{equation}
J \equiv \left( \begin{matrix}  0_{n \times n} & \mathds{1}_{n \times n} \\
- \mathds{1}_{n \times n} & 0_{n \times n} \end{matrix} \right),
\end{equation}
one defines for $\vec{u},\vec{v} \in \Z_d^{2n}$,
\begin{equation}\label{sip}
[\vec{u},\vec{v}] \equiv \vec{u}^T J \vec{v}.
\end{equation}
It is readily verified that $[\cdot, \cdot]$ is skew-symmetric and non-degenerate. As such, $[\vec{u},\vec{v}]$ defines a symplectic inner product and  $(\zx_d)^{2n}$ can be understood as a discrete phase space. An element  $S \in \mathrm{End}(\Z_d^{2n})$ is said to be {\em symplectic} if it preserves the symplectic inner product, i.e., if for each $\vec{u},\vec{v} \in \Z_d^{2n}$, $[S\vec{u},S\vec{v}] = [\vec{u},\vec{v}]$. The collection of all such elements forms the symplectic group, denoted by $\mathrm{Sp}(2n, \zx_d)$. 
Furthermore, elements $D_{\vec a}  \in \mathrm{End}(\Z_d^{2n})$ such that $\forall \vec{u} \in \Z_d^{2n}$, $D_{\vec a} \vec{u} = \vec{u} + \vec{a}$ where $\vec{a} \in \Z_d^{2n}$ are said to be phase-space displacements, and the collection of all such elements forms the group $\Z_d^{2n}$. Combinations of the latter two sorts of elements form the group $\mathrm{Sp}(2n, \zx_d) \ltimes \Z_d^{2n}$, which we term the  symplectic affine group.

Following Gross's work on the discrete Wigner representation, one can represent elements of the computational basis as probability distributations on the discrete phase space $\zx_d^{2n}$ and elements of the Clifford group as elements of $\mathrm{Sp}(2n, \zx_d) \ltimes \Z_d^{2n}$ acting thereon~\cite{Gross}. This allows one to define a symplectic representation of quopit Clifford circuits.

The discrete Wigner transformation associates to each density operator $\hat{\rho} \in B((\C^d)^{\otimes n})$ the quasi-probability distribution on $\zx_d^{2n}$ defined by 
\begin{equation}
W_{\hat{\rho}}(\vec{q},\vec{p}) = \frac{1}{d^n} \sum_{\vec{x} \in \zx_d^n} \overline{\chi}(\vec{x} \cdot \vec{p}) \langle \vec{q} + 2^{-1} \vec{x}| \hat{\rho} | \vec{q} - 2^{-1} \vec{x} \rangle.
\end{equation}
For computational basis elements $ |\vec{q}_0 \rangle \langle \vec{q}_0 |$, a simple calculation shows that $W_{| \vec{q}_0 \rangle \langle \vec{q}_0 |}$ is the uniform distribution supported on the phase space line $\{ (\vec{q}_0,\vec{p} ): \vec{p} \in \Z_d^{n} \}$.
Similarly, for an element $|\vec{p}_0 \rangle_P \langle \vec{p}_0 | $ of the momentum basis (the eigenbasis of $\chi(\hat{p})$), $W_{|\vec{p}_0 \rangle_P \langle \vec{p}_0 | }$ is uniformly supported on the phase space line $\{ (\vec{q},\vec{p}_0 ): \vec{q} \in \Z_d^{n} \}$.

For our purposes, the most important of Gross's results is the following \cite{Gross}:
\begin{proposition} 
\label{GrossSymp}
There is a map $\mu: \mathrm{Sp}(2n, \zx_d) \ltimes \Z_d^{2n} \to C_{d,n}$ satisfying
\begin{enumerate}
 \item $\mu(S,\vec{a}) \mu(T,\vec{b}) = e^{i \theta} \mu( ST, S \vec{a} + \vec{b})$ for some $\theta$, i.e., $\mu$ is a projective representation of $\mathrm{Sp}(2n, \zx_d) \ltimes \Z_d^{2n}$,
 \item For each $\hat{U} \in C_{d,n}$ there is an $(S,\vec{a}) \in \mathrm{Sp}(2n, \zx_d) \ltimes \Z_d^{2n}$ such that $\mu(S,\vec{a}) = e^{i \theta} \hat{U}$ for some $\theta$.
 \item For any density operator $\hat{\rho}$ and any $\vec{v} \in \Z_d^{2n}$, $W_{\mu(S,\vec{a}) \hat{\rho} \mu(S,\vec{a})^\dag} (S \vec{v} + \vec{a}) = W_{\hat{\rho}}(\vec{v})$ (covariance property)
\end{enumerate}
\end{proposition}

Note that property $2$ in Proposition~\ref{GrossSymp} guarantees the existence of a symplectomorphism $(S,\vec{a})$ for every element of the Clifford group, but not necessarily its uniqueness.  Nonetheless, such uniqueness does in fact hold. 
\begin{corollary}
 To each $\hat{U}\in C_{d,n}$, there is a {\em unique} $(S,\vec{a}) \in \mathrm{Sp}(2n, \zx_d) \ltimes \Z_d^{2n}$ such that $\mu(S,\vec{a}) = e^{i \theta} \hat{U}$ for some $\theta$.
\end{corollary}
The proof of this corollary is included in Appendix A.

Therefore, given a quopit Clifford circuit $\cal{C}$, there is a sequence of symplectomorphisms of $\Z_d^{2n}$, $\Phi_{\mf C}= \{ \phi_k \}_{k=1}^N$, where each $\phi_k$ is the elementary symplectomorphism associated to one of the elementary gates composing $\cal{C}$. 
We will now see that it is possible to define generating functions for these symplectomorphisms using a theory of differential forms on the affine space $\Z_d^{2n}$.

\subsection{The discrete-time analogue of the action functional for discrete systems}
\label{genfundisc}
On first thought, one might think that one can only define symplectic structures on smooth manifolds. However, a careful examination of the material presented in Section \ref{genfunctions} shows that it was not the manifold structure itself which was important but rather the existence of an algebra of differential forms. To generalise to the symplectic vector space $\Z_d^{2n}$ it therefore suffices to construct an analogue of differential forms in this context. Fortunately, the well-known K\"ahler differentials in algebraic geometry were invented for exactly this purpose (\cite{Hartshorne} Chapter II.8)\footnote{To endow a vector space with an algebra of differential forms, it must first be endowed with further geometric structure. When defining differential forms on $\R^n$ we use that this is not just a vector space but also a smooth manifold. For the vector spaces $\Z_d^{2n}$ the appropriate geometric structure is that of an affine scheme.}. Rather then delve headfirst into the theory of K\"ahler differentials, in this Section we will instead give a concrete description which more than suffices for our purposes. For interested readers, more details about the underlying mathematical structure are provided in Appendix B.

One begins by considering the algebra of polynomials in $2n$ variables over $\Z_d$, 
\begin{equation}
\Z_d[\vec q, \vec p]\equiv\Z_d[q^{(1)},\ldots, q^{(n)}, p^{(1)},\ldots, p^{(n)}].
\end{equation}
Elements of $\Z_d[\vec q, \vec p]$ can be formally differentiated using the usual formulas for differentiating polynomial functions, except that one must remember to do all arithmetic operations modulo $d$. One can then define the algebra of K\"ahler differential forms on $\Z_d^{2n}$, denoted $\Omega(\Z_d^{2n})$, as follows.  They are $\Z_d$-linear combinations of terms of the form
\begin{equation}
 f_{i_1,\ldots, i_k, j_1,\ldots j_l} \dee q^{(i_1)} \land \cdots \land \dee q^{(i_k)} \land \dee p^{(j_1)}\land  \cdots \land \dee p^{(j_l)}, 
\end{equation}
subject to the same relations as the usual differential forms on $\R^{2n}$. One can similarly decompose the K\"ahler differentials as
\begin{equation}
\Omega(\Z_d^{2n}) = \oplus_{j=0}^{2n} \Omega^j(\Z_d^{2n}),
\end{equation}
where $\Omega^j(\Z_d^{2n})$ is the vector space of K\"ahler $j$-forms. Finally, there is a differential 
\begin{equation}
\dee: \Omega^j(\Z_d^{2n}) \to \Omega^{j+1}(\Z_d^{2n})
\end{equation}
which is defined just as it is for differential forms on $\R^{2n}$, except that the usual derivative is replaced with the formal derivative explained above.

Just as one does for $\rx^{2n}$, we can extend the symplectic inner product 
on $\Z_d^{2n}$, defined in Eq.~\eqref{sip}, by introducing a symplectic form $\omega \in \Omega^2(\Z_d^{2n})$,
\begin{equation}
\label{discsympform}
\omega = \sum_{i=1}^n \dee q^{(i)} \land \dee p^{(i)}.
\end{equation}
This form satisfies $\dee \omega=0$ and is nondegenerate in the sense outlined in Appendix B. While we will not need the latter property for this paper, we feel that the fundamental role it plays in the formulation of classical dynamics on symplectic manifolds warrants its proof for the affine symplectic spaces we are considering.

In this context, a morphism $\Z_d^{2n} \to \Z_d^{2n}$ is a function $(\vec{q},\vec{p}) \mapsto (\vec{Q},\vec{P})$ such that the components of $\vec{Q}$ and $\vec{P}$ can be written as polynomials in the components of $\vec{q}$ and $\vec{p}$ with coefficients in $\Z_d$. With this definition, we say that a morphism $(\vec{q},\vec{p}) \to (\vec{Q},\vec{P})$ is symplectic if and only if
\begin{equation}
\sum_{i=1}^n \dee q^{(i)} \land \dee p^{(i)} = \sum_{i=1}^n \dee Q^{(i)}(\vec{q}, \vec{p}) \land \dee P^{(i)}(\vec{q}, \vec{p}).
\end{equation}  
There is also a discrete canonical $1$-form $\theta = \sum_{i=1}^n p^{(i)} \dee q^{(i)}$ satisfying $-\dee \theta = \omega$. We can thus make the definition:
\begin{definition}
Given a symplectomorphism $\phi: (\vec{q}, \vec{p}) \mapsto (\vec{Q}, \vec{P})$ there is an $\varepsilon \in \Omega^1(\Z_d^{2n})$ such that
\begin{equation}
\label{ClGenFun}
\sum_{i=1}^n P^{(i)} \dee Q^{(i)} - \sum_{i=1}^N p^{(i)} \dee q^{(i)}=\varepsilon,
\end{equation}
and $\dee \varepsilon = 0$. If further $\exists \tilde{G}(\vec{q},\vec{p}) \in \Z_d[\vec q, \vec p]$ such that $\varepsilon = \dee \tilde{G}(\vec{q}, \vec{p})$ then we call $\tilde{G}(\vec{q}, \vec{p})$ a {\em generating function} associated to $\phi$. 
\end{definition}
This only defines the generating function up to addition by a constant. To remove this ambiguity we will choose the generating function to have no degree $0$ components.

Notice that there exist forms which are closed but not exact and so one cannot necessarily associate a generating function to each symplectomorphism $\phi$. We will see in Section \ref{DiscGenAction} the elementary quopit Clifford symplectomorphisms do indeed have associated generating functions.

Just as in the continuous case, it may be possible to rewrite the generating function $\tilde{G}(\vec{q},\vec{p})$ in terms of $\vec{q}$ and $\vec{Q}$. This can be done exactly when the polynomial expressions for $\vec{Q}=\vec{Q}(\vec{q},\vec{p})$ can be inverted to express $\vec{p}$ in terms of $\vec{q}$ and $\vec{Q}$. As we will see, for the affine symplectomorphism associated to elements of the elementary quopit Clifford gates, it is always possible to do this inversion. As such, from now on when we refer to the generating function, we mean $G(\vec{q}, \vec{Q}) := \tilde{G}(\vec{q},\vec{p}(\vec{q},\vec{Q}))$.

Finally, consider an $n$-quopit Clifford circuit consisting of a sequence of $N$ gates, so that the space of paths through configuration space is $\Z_d^{n(N+1)}$.
Let $\Phi=\{\phi_k\}_{k=1}^N$ be the sequence of symplectomorphisms of $\Z_d^{2n}$ associated to each gate, and denote the generating function associated to $\phi_k$ by
\begin{equation}
G_{\phi_k}(\vec{q}_{k-1},\vec{q}_{k}) \in \Z_d[\vec q_{k-1}, \vec q_k].
\end{equation}
Based on the proposed definition of action functional for discrete-time paths of discrete systems, presented in the introduction, the action functional for a quopit Clifford circuit is as follows.
\begin{definition}
\label{def:genfunctionaldiscrete}
The action functional over paths in $\Z_d^{n(N+1)}$ that is associated to the sequence $\Phi$ of symplectomorphisms of $\Z_d^{2n}$,  denoted $\gamma \mapsto S_{\Phi}(\gamma)$, is
\begin{equation}
S_{\Phi}(\vec{q}_{0},\dots,\vec{q}_{N}) \equiv \frac{2\pi}{d} \sum_{k=1}^N G_{\phi_k}(\vec{q}_{k-1},\vec{q}_{k}).
\end{equation}
\end{definition}

This definition of the action functional over discrete-time paths through the {\em discrete} configuration space $\Z_d^n$ is clearly the precise analogue of the definition of the action functional over discrete-time paths on a {\em continuous} configuration space provided in Definition~\ref{def:genfunctionalcont}.

Discrete systems admit a precise analogue of Proposition~\ref{actclasstraj} as well. 

\begin{proposition}
\label{actclasstraj_discrete}
For a sequence of symplectomorphisms $\Phi = \{\phi_k\}_{k=1}^N$, let a classical discrete-time trajectory through a discrete configuration space $(\Z_d)^{n}$ 
having initial configuration $\bt q_0$ and final configuration $\bt q_N$ be denoted $\gamma^{\rm cl}_\Phi(\bt q_0, \bt q_N)$.  The action functional of Definition~\ref{def:genfunctionaldiscrete}, evaluated on $\gamma^{\rm cl}_\Phi(\bt q_0, \bt q_N)$, yields, up to a factor of $\frac{2 \pi}{d}$, the generating function of the overall symplectomorphism $\Phi_N$;  $S_\Phi(\gamma^{\rm cl}_\Phi(\bt q_0, \bt q_N)) = \frac{2 \pi}{d} G_{\Phi_N}(\bt q_0, \bt q_N)$.
\end{proposition}

Because K\"{a}hler differential can be manipulated in precisely the same way as standard differentials, the proof follows precisely that of Proposition~\ref{actclasstraj} with the insertion of appropriate factors of $\frac{2 \pi}{d}$.  

It remains to determine the precise form of the generating functions for the gates in the generating set of a quopit Clifford circuit.


\subsection{Symplectomorphisms and generating functions for quopit Clifford gates}
\label{DiscGenAction}
In Section \ref{Symprep}, we saw how the discrete Wigner transform provides a representation of a given $n$-quopit $N$ time-step Clifford circuit $\mf C$ on the discrete phase space  $\Z_d^{2n}$ in terms of a sequence of symplectomorphisms $\Phi = \{ \phi_k\}_{k=1}^N$. Further, in Section \ref{genfundisc} we defined the action functional $S_{\Phi}$ of such a sequence in terms of generating functions of the individual elements $\phi_k$. To show that the phase functional $S_{\mf C}(\gamma)$ appearing in the path sum agrees with the action functional $S_{\Phi}(\gamma)$, it remains only to compute the symplectomorphisms and generating functions of the elementary quopit Clifford gates. 

\begin{lemma}
\label{ClElemSymp}
The following symplectomorphisms are associated to the elementary quopit Clifford gates:
\begin{itemize}
\item $\phi_{\hat{F}}:(q,p)\mapsto(-p,q)$
\item  $\phi_{\hat{R}}:(q,p)\mapsto (q,p+q-2^{-1})$
\item $\phi_{\hat{\Sigma}}:(q^{(1)},q^{(2)},p^{(1)},p^{(2)}) \mapsto(q^{(1)},q^{(1)}+q^{(2)},p^{(1)}-p^{(2)},p^{(2)})$
\end{itemize}
These will be called the elementary quopit Clifford symplectomorphisms.
\end{lemma}
\begin{proof}
These are determined by direct calculation using property $3$ of Proposition~\ref{GrossSymp}, making frequent use of the identity,
\begin{equation}
\sum_{\zeta \in \Z_d} \chi(\zeta \cdot x) = d \delta_{x,0}.
\end{equation}
\begin{itemize}
 \item For the $\hat{F}$ gate, one has,
 \be{W_{\hat{F} \hat{\rho} \hat{F}^\dag} (-p,q) &=& \frac{1}{d}\sum_x \overline{\chi}(xq) \left\langle -p + 2^{-1} x \left| \left( \frac{1}{\sqrt{d}}\sum_{a,b} \chi(ab) |a \rangle \langle b| \right) \hat{\rho} \left( \frac{1}{\sqrt{d}}\sum_{s,t} \overline{\chi}(st)|s \rangle \langle t| \right) \right| -p-2^{-1} x \right\rangle \nn
&=& \frac{1}{d^2} \sum_{x,b,s} \overline{\chi}(x(q-2^{-1}(b+s)) \overline{\chi}(p(b-s)) \langle b | \hat{\rho} | s \rangle \nn
&=& \frac{1}{d} \sum_{b,s} \overline{\chi}(p(b-s)) \delta_{q,2^{-1}(b+s)} \langle b | \hat{\rho} | s \rangle \nn
&=& \frac{1}{d} \sum_{x,s} \overline{\chi}(px) \delta_{q,2^{-1}(x+2s)} \langle x + s| \hat{\rho}| s \rangle \nn
&=& \frac{1}{d} \sum_{x} \overline{\chi}(px) \langle q+ 2^{-1} x | \hat{\rho} | q-2^{-1} x \rangle = W_{\hat{\rho}} (q,p), }
where in the second last line, we used the change of variables $x=b-s$.
\item For the $\hat{R}$ gate, one has
\begin{align} 
& W_{\hat{R} \hat{\rho} \hat{R}^\dag} (q, p+q-2^{-1}) \nonumber \\
&= \frac{1}{d} \sum_x \overline{\chi}(x(p+q-2^{-1})) \left\langle q+2^{-1}x \left| \left( \sum_s \chi(2^{-1}s(s-1)) |s \rangle \langle s| \right) \hat{\rho} \left( \sum_t \overline{\chi}(2^{-1}t(t-1))|t\rangle \langle t| \right) \right| q-2^{-1}x \right\rangle \nonumber \\
&= \frac{1}{d} \sum_{x,s,t} \overline{\chi}\left(x(p+q-2^{-1})-2^{-1}s(s-1)+2^{-1}t(t-1)\right) \left\langle q+2^{-1} x | s \rangle \langle s | \hat{\rho} | t\rangle \langle t| q-2^{-1} x \right\rangle \nonumber \\
&= \frac{1}{d} \sum_x \overline{\chi}(x p) \left\langle q+2^{-1} x \left| \hat{\rho} \right| q-2^{-1} x \right\rangle = W_{\hat{\rho}}(q,p).
\end{align}
\item For the $\hat{\Sigma}$ gate, one has
\begin{eqnarray}
 &&\mkern-60mu W_{\hat{\Sigma} \hat{\rho} \hat{\Sigma}^\dag}(q^{(1)}, q^{(1)}+q^{(2)}, p^{(1)}-p^{(2)}, p^{(2)}) \nonumber \\ &=& \frac{1}{d^2} \sum_{x^{(1)},x^{(2)}} \overline{\chi}\left(x^{(1)}(p^{(1)}-p^{(2)})+x^{(2)} p^{(2)}\right)  \big\langle q^{(1)}+2^{-1} x^{(1)},q^{(1)}+q^{(2)}+2^{-1}x^{(2)} \big|
 \nonumber\\
 &&\quad\times \left( \sum_{a,b} | a,a+b \rangle  \langle a,b | \right) \hat{\rho} \left( \sum_{s,t} | s,t \rangle \langle s,s+t | \right)
  \big| q^{(1)}-2^{-1}x^{(1)},q^{(1)}+q^{(2)}-2^{-1}x^{(2)} \big\rangle \nonumber \\
 &=& \frac{1}{d^2} \sum_{x^{(1)},x^{(2)}} \overline{\chi}\left(x^{(1)}p^{(1)} + (x^{(2)}-x^{(1)})p^{(2)}\right)  \nonumber \\
 &&\quad \times \left\langle q^{(1)}+2^{-1}x^{(1)}, q^{(2)}+2^{-1}(x^{(2)}-x^{(1)}) \middle| \hat{\rho} \middle| q^{(1)}-2^{-1}x^{(1)}, q^{(2)}-2^{-1}(x^{(2)}-x^{(1)})\right\rangle \nonumber \\
 &=& \frac{1}{d^2} \sum_{x^{(1)},\tilde{x}^{(2)}} \overline{\chi}\left(x^{(1)}p^{(1)} + \tilde{x}^{(2)}p^{(2)}\right) \left\langle q^{(1)}+2^{-1}x^{(1)}, q^{(2)}+2^{-1}\tilde{x}^{(2)} \middle| \hat{\rho} \middle| q^{(1)}-2^{-1}x^{(1)}, q^{(2)}-2^{-1}\tilde{x}^{(2)}\right\rangle \nonumber \\
 &=& W_{\hat{\rho}}(q^{(1)},q^{(2)},p^{(1)},p^{(2)}),
\end{eqnarray}
where in the third line, we have made the substitution $\tilde{x}^{(2)} = x^{(2)}-x^{(1)}$. 

\end{itemize}
\end{proof}

\begin{lemma}
\label{ClElemGen}
The following are the generating functions of the symplectomorphisms associated to the elementary quopit Clifford gates:
\begin{itemize}
\item  $G_{\phi_{\hat{F}}}(q,Q) = qQ$;
\item  $G_{\phi_{\hat{R}}}(q,Q) = 2^{-1}q(q-1)$;
\item $G_{\phi_{\hat{\Sigma}}}(q^{(1)},q^{(2)},Q^{(1)},Q^{(2)}) = 0$.
\end{itemize}
\end{lemma}
\begin{proof}
Much like Lemma \ref{CVClElemGen}, this follows by direct computation from the definition in \eq{ClGenFun}.
\end{proof}

 As before, instances of the identity gate correspond to the identity symplectomorphism $( \tt q, \tt p) \mapsto (\tt q, \tt p)$ and have generating function equal to 0, and if gates act in parallel on different systems, the generating function for the overall gate is simply the sum of the generating functions of the component gates.

 Let $\mf C$ be an $n$-system quopit Clifford circuit constituting a sequence of $N$ time-steps, and let $\Phi=\{ \phi_k\}_{k=1}^N$ denote the sequence of  symplectomorphisms of $\Z_d^{2n}$ associated to the circuit. Using the quopit Clifford symplectomorphisms described in Lemma \ref{ClElemSymp}, and Definition \ref{def:genfunctionaldiscrete}, the action functional over discrete-time paths associated to $\Phi$, denoted $\gamma \mapsto \tt S_{\Phi}(\gamma)$, is the sum of the generating functions associated to these symplectomorphisms. Specifically, Lemma  \ref{ClElemGen} implies that
\begin{equation}
 \tt S_{\Phi}(\gamma) =  \frac{2\pi}{d} \left( - \sum_{\hat{F} \ {\rm gates}} q({ \rm gate}) Q({\rm gate}) + \sum_{\hat{R} \ {\rm gates}} 2^{-1} q({\rm gate}) ( q({\rm gate}) - 1) \right),
\end{equation}
which clearly corresponds to the phase factor in \eq{quopitphasefactor}.

\begin{theorem}
\label{MainDisc}
Consider an $n$-quopit Clifford circuit $\mf C$, associated in quantum theory with a unitary $\htt U \in C_{d,n}$ 
and associated, in its classical counterpart, to a symplectomorphism $\Phi$.
The functional $\tt S_{\Phi}(\gamma)$ that specifies, via Definition~\ref{def:genfunctionaldiscrete}, the action of the discrete-time path $\gamma$
through the classical counterpart of the circuit  is precisely equal to the functional $\tt S_{\mf C}(\gamma)$ that defines the phase assigned to $\gamma$ in the sum-over-paths expression for the transition amplitude of the quantum circuit.
\end{theorem}


\section{Concluding Remarks}\label{Conclusion}

In order to evaluate whether one can express quantum circuit dynamics by a sum-over-paths expression wherein the relative phases of paths are determined by a classical action functional, we began this article with a proposal for how to define an action functional for discrete-time classical dynamics.  
Such a proposal was also made in the work of Baez and Gilliam~\cite{Baez}, who sought to provide a Lagrangian formulation of discrete-time dynamics for discrete degrees of freedom.  It is currently unclear, however, precisely how our proposal relates to theirs.  
Their approach bears many similarities to ours, in that it has its foundations in algebraic geometry and the theory of K\"ahler differentials. Indeed, their article was a significant source of inspiration for the current work. The action functionals they consider take the same mathematical form as ours, namely a sum of polynomials in variables describing adjacent time steps.  However, their action functionals are used to define a discrete-time version of the Euler-Lagrange equations, whereas we associate ours to a Hamiltonian description of the discrete-time dynamics. 
A natural strategy for future attempts to relate the two approaches is to develop a generalisation of the Legendre transform for discrete-time dynamics. 

The sum-over-paths methodology has demonstrated its utility for proving relationships between quantum complexity classes. Just as it was applied by Dawson {\em et al.} to prove upper bounds on the power of arbitrary quantum circuits, so it can be applied, using the results of this paper, to prove upper bounds on the power of Clifford circuits. For instance, as we show in \cite{GK}, Theorem \ref{penult} can be used to provide an alternative proof of (a variant of) the Gottesman-Knill Theorem for quopit Clifford circuits, which states that Clifford circuits can be efficiently simulated by a classical computer \cite{gottesman1998heisenberg}. More precisely, the variant that we are able to prove in \cite{GK} states that Clifford circuits can be efficiently simulated in the $\strong(n)$ sense\footnote{A $\strong(n)$-simulation is a special case of a $\strong(f(n))$-simulation, obtained by setting $f(n)=n$. As defined in \cite{Koh}, a $\strong(f(n))$-simulation of a set of computational tasks is a deterministic classical algorithm that on input $\langle T, I, y_{|I|}\rangle$, where $T$ is a description of a quantum circuit on $n$ registers, $I = \{i_1,\ldots, i_{f(n)}\} \subseteq [n]$ and $y_{|I|} = \{y_{i_1}, \ldots, y_{i_{f(n)}}\}$, outputs the probability $p_T^{I} (y_{i_1},\ldots, y_{i_{f(n)}})$ that the outcome $y_{|I|}$ is observed.}, where a $\strong(n)$-simulation is defined to be a deterministic classical algorithm which takes as input a pair $\langle T, y\rangle$, where $T$ is a description of a quantum circuit on $n$ registers and $y \in \mathbb F_p^n$ is a length $n$-string, and outputs the probability $p_T (y)$ that the outcome $y$ is observed \cite{Koh}. Note that this result is weaker than the strong form of the Gottesman-Knill Theorem \cite{van2010classical}, which requires that the classical simulation compute not just the joint probabilities, but also any arbitrary marginal probabilities.

One can also consider other restricted models of quantum circuits, like matchgate circuits, which are circuits consisting of a certain class of two-qubit gates. Such circuits have been shown to be classically simulable under particular conditions \cite{Valiant, Jozsa, Brod}. An open question raised in \cite{Brod} is whether 
one can understand the classical simulability of such circuits via a classical hidden-variable model.  In light of our work, a further question that can be asked is whether such circuits can be described in terms of a classical action, if one describes such circuits using the sum-over-paths formulation.
 

 \color{black}

 When considering discrete variables, we have focused on Clifford circuits for systems having dimensions which are primes larger than $2$. As {\em qubit} Clifford circuits are the most familiar for those in the quantum computing community, it is natural to ask to what extent our results extend to this case.  The elementary gates for qubit Clifford circuits are
 \be{
\hat{H} &=& \frac{1}{\sqrt{2}} \sum_{q,q' \in \Z_2} (-1)^{qq'} |q \rangle \langle q'|, \nn
\hat{S} &=& \sum_{q \in \Z_2} i^q |q \rangle \langle q|, \nn
\CNOT &=& \sum_{q,q' \in \Z_2} |q,q+q' \rangle \langle q,q'|,
}
which are known as the Hadamard, Phase and CNOT gates, respectively. Because these gates are balanced, one can write a sum-over-paths expression for the transition amplitudes as in Eq.~\eqref{PathSumExpression}. One can readily see that the phase factor of a path through such a circuit is of the form $$e^{\frac{2 \pi i}{4} S(\gamma)}$$ where
\begin{equation}\label{Sforqubit}
S(\gamma) = \sum_{\hat{H} \ {\rm gates}} 2 q({\rm gate}) Q({\rm gate}) + \sum_{\hat{S} \ {\rm gates}} q({\rm gate}) 
\end{equation}
is a polynomial with coefficients in the ring $\Z_4$. The surprise here is that this is a polynomial over $\Z_4$ rather than $\Z_2$.  
It is surprising because the classical configuration variable associated to the computation basis of a qubit, which defines the space of paths in the path-sum, takes its values in $\Z_2$ rather than  $\Z_4$.

 Related to this fact, there is a significant obstruction to understanding $S(\gamma)$ as an action functional for a classical counterpart to qubit Clifford circuits. For quopit Clifford circuits, we obtained the classical counterpart by looking at Gross's discrete Wigner representation~\cite{Gross}. The property of this representation that we exploited is its so-called {\em Clifford covariance}, proven by Gross and recalled in our Proposition~\ref{GrossSymp}. It is due to this covariance property that one can represent the elementary gates by symplectomorphisms of the discrete phase space. 
However, the Wigner representation introduced by Gross is not defined for qubit systems. And while there are alternate approaches to Wigner representations of qubit systems, such as those introduced by Gibbons {\em et al.}~\cite{Wootters}, these representations are only required to be covariant under phase-space displacements, that is, unitaries generated by the Weyl operators $\{ e^{ip\hat{q}}: p \in \mathbb{R}\}$ and $\{ e^{iq\hat{p}}: q \in \mathbb{R}\}$,
 rather than the full Clifford group. Covariance under the full Clifford group is a very strong requirement to place on a Wigner representation. For $n$ quopit systems, where $d$ denotes the dimension, there are $(d^n)^{d^n +1}$ distinct Wigner representations that are covariant under phase-space displacements, but Gross's Wigner representation is the unique representation which is Clifford covariant. For qubits, it is an even stronger requirement; indeed, it has recently been shown by Zhu~\cite{Zhu} that it is {\em impossible} to define a Clifford-covariant Wigner representation for qubits. 

This rules out using a Wigner representation of qubit Clifford circuits to determine their classical counterpart. While we have no alternate proposal at present, the algebra-geometry correspondence gives some hint as to what the structure of the discrete phase space must be. In particular, since $S(\gamma)$ is a polynomial with coefficients in $\Z_4$, the discrete phase space should be some space over $\Z_4$. Further, since in the sum over paths one only sums over the 
 amplitude 
on paths with configuration variables that take their values in $\Z_2$, 
the space will not simply be $(\Z_4)^{2n}$, but rather something more complicated such as a subspace or quotient thereof. Indeed, phase spaces of this sort have been previously considered in the literature~\cite{Blasiak}, . Wallman and Bartlett~\cite{BartWall} define a positive quasi-probability representation of single qubit Clifford circuits having underlying phase space
\begin{equation}
 \{(q,p) \ | \ q\in \{0,2\}\subset \Z_4, \ p \in \Z_4\}.
\end{equation}
Unfortunately the permutation underlying the phase gate cannot be written as a polynomial map of $(q,p)$ and hence it is not possible to apply the techniques presented in this paper. Nonetheless, exploring spaces over $\Z_4$ as phase spaces for qubit Clifford circuits does suggest a new line of inquiry, 
namely, for each $n\geq 0$ one could look for a symplectic space over $\Z_4$ carrying an action of the $n$-qubit Clifford group such that the associated action functional is the one appearing in the expression for the amplitude of a path.

\black



\section*{Acknowledgements} 

RWS acknowledges Stephen Bartlett for useful discussions at an early stage of this project.  The research was begun while DEK and MDP were completing the Perimeter Scholars International Masters program at Perimeter Institute and benefitted from subsequent visits, in particular, on the occasion of the 2015 Convergence conference. \black Research at Perimeter Institute is supported in part by the Government of Canada through NSERC and by the Province of Ontario through MRI. 

DEK is supported by the National Science Scholarship from the Agency for Science, Technology and Research (A*STAR).  
MDP is supported by NSERC through the Doctoral Postgraduate Scholarship and by Corpus Christi College, Oxford.

\newpage
\appendix

\section{Proof of Corollary 1}

The proof proceeds by assuming that there are two distinct symplectomorphisms, $(S, \vec{a})$ and $(T, \vec{b})$, associated to the same unitary $\hat{U} \in C_{d,n} $, such that $\hat{U}=\mu(S,\vec{a}) = e^{i \theta} \mu(T, \vec{b})$, and deriving a contradiction.

Note that if the unitary $\hat{V}$ is associated to the symplectomorphism $(T, \vec{b})$, then the inverse unitary, $\hat{V}^{-1}$, is associated to the inverse of the latter, $(T^{-1}, -T^{-1} \vec{b})$. Note also that composition of unitaries corresponds to composition of the symplectomorphisms.  

Consider the symplectomorphism $(T^{-1}, -T^{-1}\vec{b})\circ (S, \vec{a})$.  The associated unitary is
 \begin{align}\label{inverse}
 \mu\left((T^{-1}, -T^{-1}\vec{b})\circ (S, \vec{a})\right) \propto \mu(T^{-1}, -T^{-1} \vec{b})\; \mu(S, \vec{a})\propto  \hat{U}^{\dag} \hat{U} \propto \mathds{1},
 \end{align}
 where $\propto$ denotes that the two sides differ by a unit complex number. Property $3$ of Proposition~\ref{GrossSymp} implies the identity
 \begin{align}
 W_{\mu\left((T^{-1}, -T^{-1}\vec{b})\circ (S, \vec{a})\right) \hat{\rho} \mu\left((T^{-1}, -T^{-1}\vec{b})\circ (S, \vec{a})\right) } \left(  (T^{-1}, -T^{-1}\vec{b})\circ (S, \vec{a})  [\vec{v}] \right)= 
W_{\hat{\rho}}(\vec{v}).
 \end{align}
Combining this with Eq.~\eqref{inverse}, we infer that
 \begin{align}\label{WW}
W_{\hat{\rho}} \left(  (T^{-1}, -T^{-1}\vec{b})\circ (S, \vec{a})  [\vec{v}] \right)=  W_{\hat{\rho}}(\vec{v}).
 \end{align}

Given that $(S, \vec{a})$ and $(T, \vec{b})$ are by assumption distinct, the element $(T^{-1}, -T^{-1}\vec{b})\circ (S, \vec{a})$ is distinct from the identity symplectomorphism.  It follows that there exists a $\vec{v} \in \Z_d^{2n}$ such that 
\begin{align}\label{vvv}
\vec{v} \ne \vec{v}^{\;\prime} \equiv (T^{-1}, -T^{-1}\vec{b})\circ (S, \vec{a})[\vec{v}].
\end{align}
But for any pair of vectors, $(\vec{v},\vec{v}^{\;\prime})$ such that $\vec{v} \ne \vec{v}^{\;\prime}$, there exists a $\hat{\rho}$ such that 
\begin{align}\label{Wvvv}
W_{\hat{\rho}} (\vec{v}^{\;\prime}) \ne W_{\hat{\rho}} (\vec{v}).
\end{align}
Letting $\vec{v} \equiv (q_1,\dots, q_n, p_1,\dots, p_n)$, and similarly for $\vec{v}'$, we see this as follows.  If $q_i \ne q'_i$, then take $\hat{\rho}$ to be an eigenstate of $\chi(\hat{q_i})$ with eigenvalue $\chi(q_i)$, in which case, $W_{\hat{\rho}} (\vec{v}) \ne 0$ while $W_{\hat{\rho}} (\vec{v}^{\;\prime}) = 0$.  Similarly, if $p_i \ne p'_i$, then take $\hat{\rho}$ to be an eigenstate of $\chi(\hat{p_i})$ with eigenvalue $\chi(p_i)$. 

Eqs.~\eqref{Wvvv} and \eqref{WW} provide our contradiction.

\section{K\"ahler differential forms}
In this section we will first discuss some further mathematical details of the K\"ahler differential forms introduced in Section \ref{genfundisc}, and we will then prove that the discrete symplectic form \eq{discsympform} is nondegenerate. The material on K\"ahler differentials covered here is well-known in the algebraic geometry literature (\cite{Hartshorne}, Chapter II.8) but we will nonetheless review the basic structures for the benefit of readers who are likely unfamiliar with this subject.

K\"ahler differentials can be defined by any commutative algebra $B$ over a field $\kx$. The case considered in the main text are when $\kx=\Z_d$ and $B=\Z_d[\vec{q}, \vec{p}]$ but for now we will not restrict to this case.

To begin, one first defines K\"ahler $1$-forms, $\Omega^1(B)$. These will be a module over $B$, which is essentially to say that $\Omega^1(B)$ will be a vector space where the scalars come from $B$ instead of just $\kx$.

\begin{definition}
\label{K1form}
 The module $\Omega^1(B)$ of K\"ahler $1$-forms is the $B$-module generated by elements of the form $\{ \dee b \ | \ b \in B\}$ where for each $b_1,b_2 \in B$ and $\lambda \in \kx$, the following relations hold
 \begin{enumerate}
  \item $\dee (b_1 + b_2) = \dee b_1 + \dee b_2$
  \item $\dee (b_1 b_2) = \dee(b_1) b_2 + b_1 \dee b_2$
  \item $\dee \lambda = 0$.
 \end{enumerate}
\end{definition}

\begin{definition}
 The algebra of K\"ahler differential forms $\Omega(B)$ is defined to be the exterior algebra on $\Omega^1(B)$. More explicitly, $\Omega(B)$ consists of $B$-linear combinations of terms of the form
 \begin{equation}
  \dee b_1 \land \cdots \land \dee b_j, \quad b_i \in B,
 \end{equation}
subject to the relations generated by those in Definition \ref{K1form} as well as those generated by for each $b_1, b_2 \in B$,
\begin{enumerate}
 \item  $\dee b_1 \land \dee b_2 = - \dee b_2 \land \dee b_1$
 \item $\dee b_1 \land \dee b_1 = 0$.
\end{enumerate}
\end{definition}

The K\"ahler differential forms are not just an algebra, but indeed a commutative differential graded algebra (CDGA). 
\begin{definition}
Let $B$ be an algebra over a field $\kx$. We say that $B$ is graded if it has a decomposition into vector spaces
\begin{equation}
B = B_0 \oplus B_1 \oplus B_2 \oplus \cdots,
\end{equation}
such that if $b_i \in B_i$ and $b_j \in B_j$ then $b_i b_j \in B_{i+j}$. $B$ is called commutative if $b_i b_j = (-1)^{ij} b_j b_i$.
\end{definition}

The grading on the K\"ahler differentials comes from the decomposition $\Omega(B) = \oplus_{j=0}^\infty \Omega^j(B)$, where $\Omega^j(B)$ is the vector space of $j$-forms in the usual sense.

\begin{definition}
A differential graded algebra (DGA) over a field $\kx$ is a graded algebra $B$ over $\kx$ equipped with a differential $\dee:B \to B$ such that for $b_i \in B_i$ and $b_j \in B_j$,
\begin{enumerate}
\item $\dee$ is $\kx$-linear
\item $\dee (b_i) \in B_{i+1}$
\item $\dee^2 = 0$
\item $\dee(b_i b_j) = \dee (b_i) b_j + (-1)^i b_i \dee(b_j)$
\end{enumerate}
\end{definition}

The differential on the K\"ahler differentials is defined by
\begin{equation}
 \dee ( b_0 \dee b_1 \land \cdots \land b_j) = \dee b_0 \land \dee b_1 \land \cdots \land \dee b_j
\end{equation}
and extending $\kx$-linearly.

Returning to the case when $\kx= \Z_d$ and $B= \Z_d[\vec{q}, \vec{p}]$, since $\dee$ is $\Z_d$-linear and satisfies the Leibniz rule, for any $f \in \Z_d[\vec{q}, \vec{p}]$, 
\begin{equation}
 \dee f = \sum_{i=1^n} \frac{ \partial f}{\partial q^{(i)}} \dee q^{(i)} + \sum_{i=1^n} \frac{ \partial f}{\partial p^{(i)}} \dee p^{(i)},
\end{equation}
where by $\frac{ \partial f}{\partial q^{(i)}}$ and $\frac{ \partial f}{\partial p^{(i)}}$ the formal derivative is meant. Thus, the claim in Section \ref{genfundisc} that K\"ahler differential forms on the affine space $\Z_d^{2n}$ can be manipulated according to the familiar rules for differential forms on $\R^{2n}$ holds.

We now turn to the question of whether or not the symplectic form \eq{discsympform} is non-degenerate. Recall that a $2$-form $\varepsilon$ on $\R^{2n}$ is a skew-symmetric form that takes two vector fields and outputs a smooth function. We say that such a form is non-degenerate if the map
\begin{equation}
\Gamma(T\R^{2n}) \ni X \mapsto \ \varepsilon(X,-) \in \Omega^1(\R^{2n})
\end{equation}
 is an isomorphism. For our affine space $\Z_d^{2n}$, we wish to define a notion of non-degeneracy for elements of $\Omega^2(\Z_d^{2n})$.

 First, recall that a vector field on $\R^{2n}$ is a $\R$-linear map $v: C^\infty(\R^{2n}) \to C^\infty(\R^{2n})$ which satisfies for each $f_1, f_2 \in C^\infty (\R^{2n})$, $v(f_1 f_2) = v(f_1) f_2 + f_1 v(f_2)$. Analogously we define
\begin{definition}
A $\Z_d$-linear map $D: \Z_d[\vec{q}, \vec{p}]\to \Z_d[\vec{q}, \vec{p}]$ is said to be a derivation if for each $f_1,f_2 \in \Z_d[\vec{q}, \vec{p}]$, $v(f_1 f_2) = v(f_1) f_2 + f_1 v(f_2)$. The space of all such derivations, $\mathrm{Vec}(\Z_d^{2n})$, forms an $\Z_d[\vec{q}, \vec{p}]$-module which we call the algebraic vector fields on $\Z_d^{2n}$.
\end{definition}
Next, we define a pairing between algebraic $1$-forms and vector fields by
 \begin{equation}
\langle \cdot, \cdot \rangle : \Omega^1(\Z_d^{2n}) \times \mathrm{Vec}(\Z_d^{2n}) \to \Z_d[\vec{q}, \vec{p}], \quad \langle \dee a, v \rangle = v(a),
\end{equation} 
and extend by $\Z_d[\vec{q}, \vec{p}]$-linearity in each argument. That is for $f_i, g_i, h_j \in \Z_d[\vec{q}, \vec{p}]$ and $v_j \in \mathrm{Vec}(\Z_d^{2n})$, 
\begin{equation}
\left\langle \sum_i f_i \dee g_i, \sum_j h_j v_j \right\rangle = \sum_{i,j} f_i h_j v_j(g_i).
\end{equation}
Using this pairing we reconsider elements of $\Omega^2(\Z_d^{2n})$ as skew-symmetric $\Z_d[\vec{q}, \vec{p}]$-bilinear forms on $\mathrm{Vec}(\Z_d^{2n})$. Given $ \varepsilon = \sum_i f_i \dee g_i\land \dee h_i$ and $u,v \in \mathrm{Vec}(\Z_d^{2n})$, we define
\begin{equation}
\varepsilon(u,v) = \sum_i f_i (u(g_i)v(h_i) - u(h_i) v(g_i)) = -\varepsilon(v,u).
\end{equation}

\begin{definition}
Let  $\varepsilon \in \Omega^2(\Z_d^{2n})$. $\varepsilon$ is non-degenerate if the map $\mathrm{Vec}(\Z_d^{2n}) \ni u \to \varepsilon(u, -)$ is an isomorphism of $\Z_d[\vec{q}, \vec{p}]$-modules. 
\end{definition}

\begin{theorem}
 The symplectic form $\omega = \sum_{i=1}^n \dee q^{(i)} \land \dee p^{(i)} \in \Omega^2(\Z_d^{2n})$ is nondegenerate.
\end{theorem}
\begin{proof}
 Let $u,v \in \mathrm{Vec}(\Z_d^{2n})$ be such that $\omega(u,-) = \omega(v,-)$. This implies that
\begin{equation}
\sum_{k=1}^n \left(u(q^{(k)}) \dee p^{(k)} - u(p^{(k)}) \dee q^{(k)} \right) = \sum_{k=1}^n \left(v(q^{(k)}) \dee p^{(k)} - v(p^{(k)}) \dee q^{(k)} \right).
\end{equation}
However, since $\Omega^1(\Z_d^{2n})$ is a free $\Z_d[\vec{q}, \vec{p}]$-module generated by $\{\dee q^{(1)},\ldots, \dee q^{(n)}, \dee p^{(1)},\ldots, \dee p^{(n)}\}$  the above equality implies that $u(q^{(k)})=v(q^{(k)})$ and $u(p^{(k)})=v(p^{(k)})$. Since $\Z_d[\vec{q}, \vec{p}]$ is freely generated as an algebra by the $q^{(k)}$ and $p^{(k)}$, $u=v$. We have thus proven injectivity.

Let $\beta \in \Omega^1(\Z_d^{2n})$. $\beta$ may be written as
\begin{equation}
\beta = \sum_{k=1}^n \left(\beta^1_k \dee q^{(k)} + \beta^2_k \dee p^{(k)}\right).
\end{equation}
Let $u$ be the derivation defined as $u(q^{(k)}) =  \beta^2_k$ and $u(p^{(k)}) =-\beta^1_k$. Then clearly $\omega(u,-) = \beta$, and hence it is surjective.
\end{proof}

\bibliographystyle{plain}
\bibliographystyle{unsrt}
\bibliography{Bib}

\begin{thebibliography}{10}

\bibitem{FeynmanHibbs}
Richard~P. Feynman and A.~R. Hibbs.
\newblock {\em Quantum Mechanics and Path Integrals}.
\newblock McGraw-Hill Companies, 1965.

\bibitem{Balian}
R.~Balian and C.~Bloch.
\newblock Solution of the schrödinger equation in terms of classical paths.
\newblock {\em Annals of Physics}, 85(2):514 -- 545, 1974.

\bibitem{Behtash}
Alireza Behtash, Gerald~V. Dunne, Thomas Sch\"afer, Tin Sulejmanpasic, and
  Mithat \"Unsal.
\newblock Complexified path integrals, exact saddles, and supersymmetry.
\newblock {\em Phys. Rev. Lett.}, 116:011601, Jan 2016.

\bibitem{Hormander}
Lars H{\"o}rmander.
\newblock {\em The analysis of linear partial differential operators. {I}}.
\newblock Classics in Mathematics. Springer-Verlag, Berlin, 2003.
\newblock Distribution theory and Fourier analysis, Reprint of the second
  (1990) edition [Springer, Berlin; MR1065993 (91m:35001a)].

\bibitem{Dawson}
Christopher~M. Dawson, Henry~L. Haselgrove, Andrew~P. Hines, Duncan Mortimer,
  Michael~A. Nielsen, and Tobias~J. Osborne.
\newblock Quantum computing and polynomial equations over the finite field
  $\mathbb{Z}_2$.
\newblock {\em Quantum Information and Computation}, 5(2):102--112, 2005.

\bibitem{AlgCircuits}
D.~Bacon, W.~Van~Dam, and A.~Russell.
\newblock Analyzing algebraic quantum circuits using exponential sums.
\newblock {\em unpublished. Available at http://www.cs.ucsb.edu/\textasciitilde
  vandam/publications.html}, 2008.

\bibitem{GK}
Dax~Enshan Koh, Mark~D. Penney, and Robert~W. Spekkens.
\newblock Computing quopit {C}lifford circuit amplitudes by the sum-over-paths
  technique.
\newblock arXiv:1702.03316, 2017.

\bibitem{Gottesman96}
Daniel Gottesman.
\newblock Class of quantum error-correcting codes saturating the quantum
  {H}amming bound.
\newblock {\em Phys. Rev. A}, 54:1862--1868, Sep 1996.

\bibitem{Barnes}
Richard~L. Barnes.
\newblock Stabilizer codes for continuous-variable quantum error correction.
\newblock arXiv:quant-ph/0405064, 2004.

\bibitem{GottesmanCliff}
Daniel Gottesman.
\newblock Fault-tolerant quantum computation with higher-dimensional systems.
\newblock {\em Chaos Solitons Fractals}, 10(10):1749--1758, 1999.

\bibitem{Hostens}
Erik Hostens, Jeroen Dehaene, and Bart De~Moor.
\newblock Stabilizer states and {C}lifford operations for systems of arbitrary
  dimensions and modular arithmetic.
\newblock {\em Phys. Rev. A}, 71:042315, Apr 2005.

\bibitem{Arnold}
V.~I. Arnold.
\newblock {\em {Mathematical Methods of Classical Mechanics (Graduate Texts in
  Mathematics, Vol. 60)}}.
\newblock Springer, 2nd edition, September 1997.

\bibitem{Hartshorne}
Robin Hartshorne.
\newblock {\em Algebraic geometry}.
\newblock Graduate texts in mathematics. Springer, New York, 1977.

\bibitem{Baez}
John~C. Baez and James~W. Gilliam.
\newblock An algebraic approach to discrete mechanics.
\newblock {\em Letters in Mathematical Physics}, 31(3):205--212, 1994.

\bibitem{Bartlett2012}
Stephen~D. Bartlett, Terry Rudolph, and Robert~W. Spekkens.
\newblock Reconstruction of {G}aussian quantum mechanics from {L}iouville
  mechanics with an epistemic restriction.
\newblock {\em Phys. Rev. A}, 86:012103, Jul 2012.

\bibitem{Emerson}
J.~Emerson, V.~Veitch, M.~Howard, D.~Gottesman, A.~Hamed, C.~Ferrie, and
  D.~Gross.
\newblock Negative quasi-probability, contextuality, quantum magic and the
  power of quantum computation, 2012.
\newblock Slides of a talk given at UBC, July 2013.

\bibitem{Gross}
D.~Gross.
\newblock {H}udson's theorem for finite-dimensional quantum systems.
\newblock {\em Journal of Mathematical Physics}, 47:122107, 2006.

\bibitem{Dirac}
P.~A.~M. Dirac.
\newblock The {L}agrangian in quantum mechanics.
\newblock {\em Physikalische Zeitschrift der Sowjetunion}, 3:1, 1933.

\bibitem{fredkin1990informational}
Edward Fredkin.
\newblock An informational process based on reversible universal cellular
  automata.
\newblock {\em Physica D: Nonlinear Phenomena}, 45(1-3):254--270, 1990.

\bibitem{Toffoli}
Tommaso Toffoli and Norman Margolus.
\newblock {\em Cellular Automata Machines}.
\newblock MIT Press, 1987.

\bibitem{Wolfram}
Stephen Wolfram.
\newblock {\em A New Kind of Science}.
\newblock (Champaign IL: Wolfram Media Inc.), 2002.

\bibitem{ng1995limitation}
Y~Ng and H~van Dam.
\newblock Limitation to quantum measurements of space-time distances.
\newblock {\em Annals of the New York Academy of Sciences}, 755(1):579--584,
  1995.

\bibitem{amelino1994limits}
Giovanni Amelino-Camelia.
\newblock Limits on the measurability of space-time distances in (the
  semiclassical approximation of) quantum gravity.
\newblock {\em Modern Physics Letters A}, 9(37):3415--3422, 1994.

\bibitem{rovelli1995discreteness}
Carlo Rovelli and Lee Smolin.
\newblock Discreteness of area and volume in quantum gravity.
\newblock {\em Nuclear Physics B}, 442(3):593--619, 1995.

\bibitem{Wallden}
Petros Wallden.
\newblock Causal sets: Quantum gravity from a fundamentally discrete spacetime.
\newblock {\em Journal of Physics: Conference Series}, 222(1):012053, 2010.

\bibitem{watrous1995one}
John Watrous.
\newblock On one-dimensional quantum cellular automata.
\newblock In {\em Foundations of Computer Science, 1995. Proceedings., 36th
  Annual Symposium on}, pages 528--537. IEEE, 1995.

\bibitem{schumacher2004reversible}
Benjamin Schumacher and Reinhard~F Werner.
\newblock Reversible quantum cellular automata.
\newblock {\em arXiv preprint quant-ph/0405174}, 2004.

\bibitem{d2017quantum}
Giacomo~Mauro DAriano and Paolo Perinotti.
\newblock Quantum cellular automata and free quantum field theory.
\newblock {\em Frontiers of Physics}, 12(1):120301, 2017.

\bibitem{Bel64a}
J.~S. Bell.
\newblock On the {Einstein-Podolsky-Rosen} paradox.
\newblock {\em Physics}, 1:195, 1964.

\bibitem{Bel66a}
J.~S. Bell.
\newblock On the problem of hidden variables in quantum mechanics.
\newblock {\em Rev. Mod. Phys.}, 38:447--452, 1966.

\bibitem{KS}
Simon Kochen and Ernst Specker.
\newblock The problem of hidden variables in quantum mechanics.
\newblock {\em Journal of Applied Mathematics and Mechanics}, 17:59--87, 1967.

\bibitem{Spekkens05}
Robert~W Spekkens.
\newblock Contextuality for preparations, transformations, and unsharp
  measurements.
\newblock {\em Physical Review A}, 71(5):052108, 2005.

\bibitem{Spekkens08}
Robert~W Spekkens.
\newblock Negativity and contextuality are equivalent notions of
  nonclassicality.
\newblock {\em Physical review letters}, 101(2):020401, 2008.

\bibitem{FerrieEmerson}
Christopher Ferrie and Joseph Emerson.
\newblock Frame representations of quantum mechanics and the necessity of
  negativity in quasi-probability representations.
\newblock {\em Journal of Physics A: Mathematical and Theoretical},
  41(35):352001, 2008.

\bibitem{Kent}
Adrian Kent.
\newblock Path integrals and reality.
\newblock arXiv:1305.6565, 2013.

\bibitem{Barrett}
Jonathan Barrett, Lucien Hardy, and Adrian Kent.
\newblock No signaling and quantum key distribution.
\newblock {\em Phys. Rev. Lett.}, 95:010503, Jun 2005.

\bibitem{Pironio}
S.~Pironio, A.~Ac{\'\i}n, S.~Massar, A.~Boyer de~la Giroday, D.~N. Matsukevich,
  P.~Maunz, S.~Olmschenk, D.~Hayes, L.~Luo, T.~A. Manning, and C.~Monroe.
\newblock Random numbers certified by bell's theorem.
\newblock {\em Nature}, 464(7291):1021--1024, 04 2010.

\bibitem{Anders}
Janet Anders and Dan~E. Browne.
\newblock Computational power of correlations.
\newblock {\em Phys. Rev. Lett.}, 102:050502, Feb 2009.

\bibitem{Raussendorf}
Robert Raussendorf, Daniel~E. Browne, and Hans~J. Briegel.
\newblock Measurement-based quantum computation on cluster states.
\newblock {\em Phys. Rev. A}, 68:022312, Aug 2003.

\bibitem{HowardEtAl}
Mark Howard, Joel Wallman, Victor Veitch, and Joseph Emerson.
\newblock Contextuality supplies the 'magic' for quantum computation.
\newblock {\em Nature}, 510(7505):351--355, 2014.

\bibitem{Albash}
Tameem Albash, Troels~F. R{\o}nnow, Matthias Troyer, and Daniel~A. Lidar.
\newblock Reexamining classical and quantum models for the {D}-{W}ave one
  processor.
\newblock {\em EPJ-ST}, 224:111, 2015.

\bibitem{Bartlett}
Stephen~D. Bartlett, Barry~C. Sanders, Samuel~L. Braunstein, and Kae Nemoto.
\newblock Efficient classical simulation of continuous variable quantum
  information processes.
\newblock {\em Phys. Rev. Lett.}, 88:097904, Feb 2002.

\bibitem{Clark}
Sean Clark.
\newblock Valence bond solid formalism for d-level one-way quantum computation.
\newblock {\em Journal of Physics A: Mathematical and General}, 39(11):2701,
  2006.

\bibitem{gottesman1998heisenberg}
Daniel Gottesman.
\newblock The heisenberg representation of quantum computers.
\newblock {\em arXiv preprint quant-ph/9807006}, 1998.

\bibitem{Koh}
Dax~Enshan Koh.
\newblock Further extensions of {C}lifford circuits and their classical
  simulation complexities.
\newblock arXiv:1512.07892, 2015.

\bibitem{van2010classical}
Maarten Van~den Nest.
\newblock Classical simulation of quantum computation, the gottesman-knill
  theorem, and slightly beyond.
\newblock {\em Quantum Information \& Computation}, 10(3-4):258--271, 2010.

\bibitem{Valiant}
Leslie~G Valiant.
\newblock Quantum circuits that can be simulated classically in polynomial
  time.
\newblock {\em SIAM Journal on Computing}, 31(4):1229--1254, 2002.

\bibitem{Jozsa}
Richard Jozsa and Akimasa Miyake.
\newblock Matchgates and classical simulation of quantum circuits.
\newblock {\em Proceedings of the Royal Society of London A: Mathematical,
  Physical and Engineering Sciences}, 464(2100):3089--3106, 2008.

\bibitem{Brod}
Daniel~J Brod.
\newblock Efficient classical simulation of matchgate circuits with generalized
  inputs and measurements.
\newblock {\em Physical Review A}, 93(6):062332, 2016.

\bibitem{Wootters}
Kathleen~S. Gibbons, Matthew~J. Hoffman, and William~K. Wootters.
\newblock Discrete phase space based on finite fields.
\newblock {\em Phys. Rev. A}, 70:062101, Dec 2004.

\bibitem{Zhu}
Huangjun Zhu.
\newblock Permutation symmetry determines the discrete {W}igner function.
\newblock {\em Phys. Rev. Lett.}, 116:040501, Jan 2016.

\bibitem{Blasiak}
Pawel Blasiak.
\newblock Quantum cube: A toy model of a qubit.
\newblock {\em Physics Letters A}, 377(12):847 -- 850, 2013.

\bibitem{BartWall}
Joel~J Wallman and Stephen~D Bartlett.
\newblock Non-negative subtheories and quasiprobability representations of
  qubits.
\newblock {\em Physical Review A}, 85(6):062121, 2012.

\end{thebibliography}

\end{document}